\documentclass{article}  
\newcommand{\PalindromeResult}{\ensuremath{7 (\pi^2 / 6 + 1/2) \alpha n - 5 n - 1}}
\newcommand{\RepeatResult}{\ensuremath{3(\pi^2/6 + 5/2) \alpha n}}
\newcommand{\LCEds}{$\text{LCE}^{\leftrightarrow}$ data structure}

\makeatletter
\newcommand{\CustomLabel}[2]{%
   \protected@write \@auxout {}{\string \newlabel {#1}{{#2}{\thepage}{#2}{#1}{}} }%
   \hypertarget{#1}{#2}%
}
\newcommand{\GappedCase}[2]{\noindent\textbf{#1:} #2.}
\newcommand{\GappedSubCase}[3]{\GappedCase{#1}{#2, see \Cref{#3}}}
\makeatother

\newlength{\lenTextFigureBlock}
\newcommand{\TextFigureBlock}[3]{%
	\setlength{\lenTextFigureBlock}{\linewidth-#1\linewidth}
	\begin{textminipage}{#1\linewidth}
		#2
	\end{textminipage}
	\begin{minipage}{\lenTextFigureBlock}
		\hfill
		#3
	\end{minipage}
}
\newcommand{\MyFloatBox}[2]{%
\floatbox[{\capbeside\thisfloatsetup{capbesideposition={right,center}}}]{figure}[\FBwidth]{#2}{#1}}

\usepackage{calc}
\usepackage{comment}
\usepackage{cutwin}
\usepackage{wrapfig}
\usepackage{hyperref}
\usepackage{amsfonts,amssymb,amsthm}
\usepackage[fleqn]{amsmath}
\usepackage[caption=false]{subfig} %
\usepackage{floatrow}
\usepackage{fullpage}

	\newenvironment{subproof}{\block{Sub-Proof}}{\hfill\ensuremath{\blacksquare}}
\newenvironment{subclaim}{\block{Sub-Claim}}{}

\usepackage{microtype}%

\usepackage[utf8]{inputenc}
\usepackage{multicol}
\usepackage{graphicx}
\usepackage[usenames,dvipsnames]{xcolor}

\newcommand{\UnaryOperator}[2][]{%
	\ifx&#1&%
	\ensuremath{\mathop{}\mathopen{}#2\mathopen{}}%
	\else%
	\ensuremath{\mathop{}\mathopen{}#2\mathopen{}\left(#1\right)}%
\fi%
}
\RequirePackage[capitalise]{cleveref}
\newcommand{\Int}[2]{\ensuremath{\lbrack#1..\nobreak#2\rbrack}}
\newcommand{\substr}[3]{\ensuremath{\mathop{}\mathopen{}#1\mathopen{}\left\lbrack#2..#3\right\rbrack}} %
\newcommand{\sumExp}[1]{\UnaryOperator[#1]{\mathcal{E}}}
\newcommand{\Oh}[1]{\UnaryOperator[#1]{\mathcal{O}}}

\newcommand{\abs}[1]{\left|#1\right|} %
\newcommand{\gauss}[1]{\left\lfloor#1\right\rfloor} %
\newcommand{\upgauss}[1]{\left\lceil#1\right\rceil} %
\newcommand{\tuple}[1]{\left(#1\right)} %
\newcommand{\intWort}[1]{\emph{\textbf{#1}}}
\DeclareMathAlphabet{\mathup}{OT1}{\familydefault}{m}{n}
\newcommand{\idbb}[1]{\ensuremath{\mathbb{#1}}}
\newcommand{\N}{\idbb{N}}
\newcommand{\RealNumber}{\idbb{R}}
\newcommand{\Z}{\idbb{Z}}
\newcommand{\imgRepeat}{\ensuremath{\varphi_{\textup{}}(\grGR)}}
\newcommand{\mapRepeat}{\ensuremath{\varphi_{\textup{}}}}
\newcommand{\imgPali}{\ensuremath{\varphi_{\mathup{\intercal}}(\gpNP)}}
\newcommand{\mapPali}{\ensuremath{\varphi_{\mathup{\intercal}}}}
\newcommand{\Char}[1]{\texttt{#1}}
\newcommand{\occ}          {\ensuremath{\textup{occ}}}
\DeclareRobustCommand\sitemize[1]{\begin{itemize}#1\end{itemize}}
\usepackage{suffix}
\usepackage{xstring}
\DeclareRobustCommand\scases[2]{%
	\ifmmode%
	#1\begin{cases}#2\end{cases}
	\else%
	\[#1\begin{cases}#2\end{cases}\]
	\fi%
}
\WithSuffix\DeclareRobustCommand\scases*[2]{%
	\saveexpandmode\expandarg%
	\StrSubstitute{\noexpand#2}{&}{\ }[\vv]%
	\StrSubstitute{\vv}{\empty\\}{\ #1}[\vvv]%
	\restoreexpandmode%
	\ensuremath{#1\ \vvv}}
\WithSuffix\DeclareRobustCommand\sitemize*[1]{%
	\let\normalnewline=\item
		\let\item=\relax
		#1
	\let\item=\normalnewline
}
\RequirePackage{thmtools} %

\theoremstyle{definition}
\newtheorem{theorem}{Theorem}[section]
\newtheorem{lemma}[theorem]{Lemma}
\newtheorem{corollary}[theorem]{Corollary}
\newtheorem{fact}[theorem]{Fact}
\newtheorem{example}[theorem]{Example}

\newtheorem{definition}[theorem]{Definition}

\Crefname{fact}{Fact}{Facts}
\crefname{fact}{fact}{facts}

\newcommand{\menge}[1]{\left\{#1\right\}} %
\newcommand{\block}[1]{\noindent\textbf{#1. }}

\usepackage{enumerate} %
\usepackage[shortlabels]{enumitem} %

\newcommand{\intervalI}{\mathcal{I}} %
\newcommand{\ibeg}[1]{\mathsf{b}(#1)}%
\newcommand{\iend}[1]{\mathsf{e}(#1)}%

\newcommand{\RR}[1]{\ensuremath{{#1}_{\mathup{\rho}}}}
\newcommand{\LL}[1]{\ensuremath{{#1}_{\mathup{\lambda}}}}
\newcommand{\arm}   {\ensuremath{u}}
\newcommand{\Sarm}{\ensuremath{u}}
\newcommand{\uarm}   {\ensuremath{u}}
\newcommand{\warm}   {\ensuremath{w}}
\newcommand{\subarm}{\ensuremath{s}}
\newcommand{\Ssubarm}{\ensuremath{s}}
\newcommand{\rep}   {\ensuremath{r}}
\newcommand{\gap}   {\ensuremath{v}}

\newcommand{\zvar}  {\ensuremath{z}}

\newcommand{\agr}  {\ensuremath{(\LL\arm,\RR\arm)}} %
\newcommand{\sagr}  {\ensuremath{(\s{\LL\arm},\s{\RR\arm})}} %
\newcommand{\qvar}  {\ensuremath{q}}

\newcommand{\pali}[1]{\ensuremath{{#1}^{\mathup{\intercal}}}}

\newcommand{\CoverN}[1]{\mathcal{C}_{#1}}
\newcommand{\SubCover}{\ensuremath{C}}
\newcommand{\SetE}{\ensuremath{E}}
\newcommand{\weight}[1][]{\UnaryOperator[#1]{\mathsf{w}}}

\newcommand{\corr}[1]{\mathit{c}(#1)}

\newcommand{\grGR}  [1][\warm]{\ensuremath{\mathcal{G}_\alpha(#1)}}
\newcommand{\gpP}   [1][\warm]{\ensuremath{\pali{\beta\mathcal{P}}_{\hspace{-0.2em}\alpha}(#1)}}
\newcommand{\gpNP}  [1][\warm]{\ensuremath{\pali{\overline{\beta\mathcal{P}}_{\hspace{-0.2em}\alpha}\hspace{-0.2em}}(#1)}}

\newcommand{\gpGR}  [1][\warm]{\ensuremath{\pali{\mathcal{G}}_\alpha(#1)}}
\newcommand{\grP}   [1][w]{\ensuremath{\beta\mathcal{P}_{\hspace{-0.2em}\alpha}(#1)}}
\newcommand{\grNP}  [1][w]{\ensuremath{\overline{\beta\mathcal{P}_{\hspace{-0.2em}\alpha}}(#1)}}

\newcommand{\period}{p}
\newcommand{\s}[1]{\overline{#1}}

\newenvironment{textminipage}[1]{%
	\noindent%
	\edef\myindent{\the\parindent}%
	\begin{minipage}{#1}
	\setlength{\parindent}{\myindent}
}{\end{minipage}}

\usepackage{tikz}
	\definecolor{solarizedYellow}{HTML}{B58900}
	\definecolor{solarizedOrange}{HTML}{CB4B16}
	\definecolor{solarizedRed}{HTML}{DC322F}
	\definecolor{solarizedMagenta}{HTML}{D33682}
	\definecolor{solarizedViolet}{HTML}{6C71C4}
	\definecolor{solarizedBlue}{HTML}{268BD2}
	\definecolor{solarizedCyan}{HTML}{2AA198}
	\definecolor{solarizedGreen}{HTML}{859900}
\usetikzlibrary{patterns}
\tikzstyle{arm} = [pattern=horizontal lines, pattern color=solarizedBlue!15]
\tikzstyle{subarm} = [pattern=vertical lines, pattern color=solarizedYellow!15]
\tikzstyle{gap} = [pattern=north east lines, pattern color=gray!20]
\tikzstyle{rep} = [pattern=north west lines, pattern color=solarizedRed!15]
\tikzstyle{diff} = [pattern=dots, pattern color=solarizedCyan!15]

\edef\ourCol{0}
\newcount\ourRow%
\pgfmathsetcount{\ourRow}{0}

\newcommand{\col}[3][]{%
	\draw [#1] (\ourCol,\ourRow) rectangle node [#1,fill=white,inner sep=0em] {#3} (\ourCol+#2,\ourRow-1);
\edef\ourCol{\ourCol+#2}
}
\newcommand{\void}[1]{%
\edef\ourCol{\ourCol+#1}
}
\newcommand{\newrow}{%
\pgfmathsetcount{\ourRow}{\ourRow-1}
\edef\ourCol{0}
}
\newcommand{\marke}[1]{%
\draw (\ourCol,\ourRow) -- (\ourCol,\ourRow-1) node [anchor=north] {#1};
}
\newcommand{\markeup}[1]{%
\draw (\ourCol,\ourRow-1) -- (\ourCol,\ourRow) node [anchor=south] {#1};
}
\newcommand{\dist}[2][]{%
\draw [<->] (\ourCol,\ourRow+0.5) -- (\ourCol+#2,\ourRow+0.5) node [midway,anchor=north] {#1};
\edef\ourCol{\ourCol+#2}
}
\newcommand{\distup}[2][]{%
\draw [<->] (\ourCol,\ourRow-0.5) -- (\ourCol+#2,\ourRow-0.5) node [midway,anchor=south] {#1};
\edef\ourCol{\ourCol+#2}
}

\RequirePackage[numbers,sectionbib]{natbib}
\makeatletter %
\renewcommand*{\NAT@spacechar}{~}
\makeatother

\title{Improved Upper Bounds on all Maximal $\alpha$-gapped Repeats and Palindromes}

\usepackage{authblk}
\author[1]{Tomohiro I}
\author[2]{Dominik K\"{o}ppl}%
\affil[1]{Kyushu Institute of Technology, Japan}
\affil[2]{Department of Computer Science, TU Dortmund, Germany}
\begin{document}
\maketitle

\begin{abstract}
We show that the number of all maximal $\alpha$-gapped repeats and palindromes of a word of length~$n$ is at most~$\RepeatResult$ and~$\PalindromeResult$, respectively.
\end{abstract}

\providecommand{\PalindromeResult}{\ensuremath{7 (\pi^2 / 6 + 1/2) \alpha n - 5 n - 1}}
\providecommand{\RepeatResult}{\ensuremath{3(\pi^2/6 + 5/2) \alpha n}}

\section{Introduction}\label{secIntro}
Given a word~$\warm$,
a \emph{gapped repeat} is a triple of integers~$(\LL{i}, \RR{i}, \Sarm)$ with the properties
(a)~$0 < \RR{i}-\LL{i}$, and (b)~$\warm[\LL{i}..\LL{i}+\Sarm-1] = \warm[\RR{i}..\RR{i}+\Sarm-1]$.
A variant are \emph{gapped palindromes} with the properties
(a)~$0 \le \RR{i}-\LL{i}$, and (b)~$\warm[\LL{i}..\LL{i}+\Sarm-1]$ is equal to the reverse of $\warm[\RR{i}..\RR{i}+\Sarm-1]$.
In both cases (repeats or palindromes), $\warm[\LL{i}..\LL{i}+\Sarm-1]$ and $\warm[\RR{i}..\RR{i}+\Sarm-1]$ are called \emph{left} and \emph{right arm}, respectively.
Given a real number $\alpha \ge 1$, $(\LL{i}, \RR{i}, \Sarm)$ is called \emph{$\alpha$-gapped} if $\RR{i}-\LL{i} \le \alpha\Sarm$.
A gapped repeat is \emph{maximal} 
if its arms can be extended neither to their left nor to their right sides (to form a larger gapped repeat).
Similarly, a gapped palindrome is \emph{maximal} if it can be extended neither inwards nor outwards.
Maximal $\alpha$-gapped repeats and palindromes starred in several recent papers~\cite{gappedPalindroms,KolpakovPPK14,crochemore16optimal,gawrychowski18tighter}.
The most intriguing questions are: 
\begin{enumerate}
	\item How to compute all maximal $\alpha$-gapped repeats/palindromes efficiently, and:
	\item What is the maximum number of maximal $\alpha$-gapped repeats/palindromes in a word?
\end{enumerate}
Previously, the second question was answered with \Oh{\alpha^2 n}~\cite{gappedPalindroms,KolpakovPPK14},
subsequently with \Oh{\alpha n}~\cite{crochemore16optimal}, and finally with $18\alpha n$ and $28\alpha n + 7n$ for maximal $\alpha$-gapped repeats and maximal $\alpha$-gapped palindromes, respectively~\cite{gawrychowski18tighter}.
Following this line of achievements, this article gives yet another improvement to those answers:
\begin{itemize}
	\item The number of all maximal $\alpha$-gapped repeats in a word of length $n$ is at most $\RepeatResult$~(\Cref{thmMaxReps}).
	\item The number of all maximal $\alpha$-gapped palindromes in a word of length $n$ is at most $\PalindromeResult$~(\Cref{thmMaxPal}). 
\end{itemize} 
The improvement of the upper bound on the number of all maximal $\alpha$-gapped repeats is a small refinement step (in \Cref{lemmaPoints,lemmaPointsRepPropA,lemmaPointsRep}), whereas our 
new upper bound on the number of all maximal $\alpha$-gapped palindromes involves a more thorough analysis (in \Cref{lemmaPalNPCover}).
Here, the main difference to~\cite{gawrychowski18tighter} is that 
\begin{itemize}
	\item we define a periodic gapped palindrome to have a left arm with a sufficiently long periodic suffix (instead of prefix), and that
	\item we support overlaps (previous results assumed that $\LL{i}+\Sarm \le \RR{i}$).
\end{itemize}
The former change helps us to attain a refined upper bound at the expense of a more thorough analysis.
The latter change is a generalization, since our proofs work for both supporting and prohibiting overlaps.
This generalization makes the maximality property more natural, since a left/right extension of a gapped repeat 
(resp.\ an inward extension of a gapped palindrome) is always a gapped repeat (resp.\ gapped palindrome).
\begin{example}\label{exGappedAAA}
The first two characters of $\warm = \Char{aaa}$ form a gapped repeat~$(1,2,1)$. 
The right extensions~$(1,2,2)$ of both arms is only a gapped repeat if overlaps are supported.
Similarly, $(1,3,1)$ is a gapped palindrome, but the inward extension $(1,2,2)$ is a gapped palindrome only if overlaps are supported.
\end{example}

A natural question arising from this generalization is whether we can still compute the set of all maximal $\alpha$-gapped repeats and palindromes 
within the same bounds when supporting overlaps. 
We can answer this question affirmatively in the penultimate section of this article. %
Throughout this article, we heavily borrow the notations and ideas evolved by~\citet{gawrychowski18tighter} and~\citet{KolpakovPPK14}.

\section{Preliminaries}

A (real) \intWort{interval} $\intervalI=[b,e] \subset \RealNumber$ for $b,e \in \RealNumber$ is the set of all real numbers~$i \in \RealNumber$ with $b \le i \le e$.
We write $[b,e)$, $(b,e]$ or $(b,e)$ if $e$, $b$, or both values are not included in the interval.
For an interval $\intervalI$, $\ibeg{\intervalI}$ and $\iend{\intervalI}$ denote the beginning and end of $\intervalI$, respectively.

A special kind of intervals are integer intervals $\intervalI=[b..e]$, 
where $\intervalI$ is the set of consecutive integers from $b = \ibeg{\intervalI} \in \Z$ to $e = \iend{\intervalI} \in \Z$, for $b\le e$. 
We write $\abs{\intervalI}$ to denote the length of $\intervalI$; i.e., $\abs{\intervalI}=\iend{\intervalI}-\ibeg{\intervalI}+1$.

Let $\Sigma$ be a finite alphabet; an element of $\Sigma$ is called \intWort{character}.
$\Sigma^*$ denotes the set of all finite \intWort{words} over $\Sigma$. 
The \intWort{length} of a word $\warm\in \Sigma^*$ is denoted by $\left|w\right|$. 
For $v = xuy$ with $x,u,y \in \Sigma^*$, we call $x$, $u$ and $y$ a \intWort{prefix}, \intWort{factor}, and \intWort{suffix} of $v$, respectively.
We denote by $\warm[i]$ the character occurring at position $i$ in $\warm$, and by $\substr{\warm}{i}{j}$ the factor of $\warm$ starting at position $i$ and ending at position $j,$ consisting of the catenation of the characters $\warm[i], \ldots, \warm[j],$ where $1\leq i\leq j\leq n$; $\substr{\warm}{i}{j}$ is the empty word if $i>j$.
{By ${\pali{\warm}}$ we denote the \intWort{reverse} of $\warm$.} 

The notation $\warm\Int{b}{e}$ can be ambivalent: it can denote both a factor and the occurrence of this factor starting at position~$b$ in~$\warm$.
The second entity is called the segment\footnote{This notion was coined in~\cite{crochemore16optimal}.}~$\substr{\warm}{b}{e}$: A \intWort{segment}~$\substr{\warm}{b}{e}$ of a word $\warm$ is the occurrence of a factor~$f$ equal to $\substr{\warm}{b}{e}$ in $\warm$; 
we say that $f$ \intWort{occurs} at position~$b$ in~$\warm$.
While a factor is identified only by a sequence of characters, a segment is also identified by its position in the word.
A conclusion is that segments are always unique, while a word may contain multiple occurrences of the same factor.
We use the same notation for defining factors and segments of a word.
For two segments $\uarm$ and $\s\uarm$ of a word~$\warm$, we write $\uarm \equiv \s\uarm$ 
if they start at the same position in~$\warm$ and have the same length.
We write $\uarm = \s\uarm$ if the factors identifying these segments are the same (hence $\uarm \equiv \s\uarm \Rightarrow \uarm = \s\uarm$).
We implicitly use segments both like factors of $\warm$ 
and as intervals contained in $\Int{1}{\abs{\warm}}$, e.g., 
we write $\uarm \subseteq \s\uarm$ if
two segments $\uarm := \substr{\warm}{b}{e}, \s\uarm := \substr{\warm}{\s{b}}{\s{e}}$ of $\warm$ satisfy $\Int{b}{e} \subseteq \Int{\s{b}}{\s{e}}$, i.e.,
$\ibeg{\s\uarm} \le \ibeg{\uarm} \le \iend{\uarm} \le \iend{\s\uarm}$.

A \intWort{period} of a word $\warm$ over $\Sigma$ is a positive integer $\period < \abs{\warm}$ such that $\warm[i]=\warm[j]$ for all $i$ and $j$ 
with $1 \le i,j \le \abs{\warm}$ and $i \equiv j\pmod{p}$.
A word $\warm$ whose smallest period is at most $\gauss{\abs{\warm}/2}$ is called \intWort{periodic};
otherwise, $\warm$ is called \intWort{aperiodic}.
A \intWort{repetition} in a word $\warm$ is a periodic factor;
a \intWort{run} is a maximal repetition; the \intWort{exponent} of a run is the (rational) number of times the smallest period fits in that run. 
The exponent of a run~$\rep$ is denoted by $\exp(\rep)$.
The sum of the exponents of runs in the word~$\warm$ is denoted by $\sumExp{\warm}$.
We use the following results from the literature:

\begin{lemma}[{\cite{FineWilf65}}]\label{lemmaWeekPeriodicity}
Given a word~$\warm$ with two periods~$\period$ and~$\period'$ such that $\period+\period' \le \abs{\warm}$, 
the greatest common divisor~$\gcd(\period,\period')$ of~$\period$ and~$\period'$ is also a period of~$\warm$.
\end{lemma}

\begin{figure}[ht]
	\MyFloatBox{%
				\begin{tikzpicture}[yscale=0.4,xscale=0.5]
		  \col[arm]{6}{$\LL\uarm$}
		  \newrow
		  \void{1}
		  \col[arm]{6}{$\RR\uarm$}
		  \newrow
		  \dist[$\delta$]{1}
		  \void{2}
		  \col[subarm]{4}{prefix of $\uarm$}
		  \newrow
		  \dist[$\period$]{3}
		  \newrow
		  \void{1}
		  \dist[$\period-\delta$]{2}
		  \dist[$>\period$]{4}
			\end{tikzpicture}
		}{%
		\caption{Setting of the proof of \Cref{lemmaRepetitiveOccs} with $\delta < p$. There are two occurrences~$\LL\uarm$ and $\RR\uarm$ of~$\uarm$ with an overlap of $2p-\delta$ characters.
			Both occurrences induce a run of period~$\delta$.
			There are at least three occurrences of $\uarm$'s prefix of length~$p+1$ (starting at $\ibeg{\LL\uarm}, \ibeg{\RR\uarm}$, and $\ibeg{\LL\uarm}+p$).
		}
		\label{figGappedRepetitiveOccs}
}%
\end{figure}
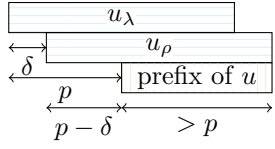

\begin{corollary}\label{lemmaRepetitiveOccs}
A periodic factor~$\uarm$ in a word~$\warm$ with the smallest period~$\period$ cannot have two distinct occurrences~$\LL\uarm$ and $\RR\uarm$ in~$\warm$ with
$\abs{\ibeg{\LL\uarm} - \ibeg{\RR\uarm}} < \period$.
\end{corollary}
\begin{proof}
	Since the smallest period of~$\uarm$ is~$\period$, $\abs{\uarm} > 2 \period$ holds.
	Assume for a contradiction that two distinct occurrences $\LL\uarm$ and~$\RR\uarm$ of $\uarm$ exist in~$\warm$ with
	a distance $\delta := \ibeg{\RR\uarm} - \ibeg{\LL\uarm}$ such that $0 < \delta < \period$ (see also \Cref{figGappedRepetitiveOccs}).
	Since $\abs{\LL\uarm \cap \RR\uarm} \ge 2p-\delta \ge p$, $\delta$ is a period of $\uarm$.
	Additionally, since~$\uarm$ has the smallest period~$\period$, 
	there is another occurrence of a prefix of~$\uarm$ starting at~$\ibeg{\LL\uarm}+\period-\delta$ with a length of at least~$\period+\delta > \period$.
	Hence, $\period-\delta$ is also a period of~$\uarm$.
	Because the sum of both periods~$\delta$ and~$\period-\delta$ is less than $\abs{\uarm}$, \Cref{lemmaWeekPeriodicity} states that
	$\gcd(\delta,\period-\delta) < \period$ is a period of~$\uarm$.
	This contradicts the fact that~$\period$ is the smallest period of~$\uarm$.
\end{proof}

\begin{lemma}[\cite{runstheorem}]\label{lemmaExponent}
	For a word $\warm$, $\sumExp{\warm} < 3\abs{\warm}$.
\end{lemma}

Instead of working with triples of integers~$(\LL{i}, \RR{i}, \Sarm)$ as in \Cref{secIntro} when representing gapped repeats and palindromes, 
we stick to pairs of segments~$\tuple{\substr{\warm}{\LL{i}}{\LL{i}+\Sarm-1},\substr{\warm}{\RR{i}}{\RR{i}+\Sarm-1}}$ for convenience:
For a word $\warm$, we call a pair of segments $(\LL\arm,\RR\arm)$ a \intWort{gapped repeat} (resp.\ \intWort{gapped palindrome})
with \intWort{period} $\qvar = \ibeg{\RR\arm} - \ibeg{\LL\arm}$ iff 
\begin{itemize}
	\item $\ibeg{\LL\arm}+1 \le \ibeg{\RR\arm}$ and $\RR\arm = \LL\arm$ in the case of a gapped repeat, or
	\item $\ibeg{\LL\arm} \le \ibeg{\RR\arm}$ and $\RR\arm = \pali{\LL\arm}$ in the case of a gapped palindrome (it is possible that $\LL\arm \equiv \RR\arm$).
\end{itemize}
The segments $\LL\arm$ and $\RR\arm$ are called left and right \intWort{arm}, respectively.
The value~$\ibeg{\RR\arm}-\iend{\LL\arm} -1$ is called the \intWort{gap}, and is the distance between both arms in case that it is positive.
For $\alpha\geq 1$, the gapped repeat or gapped palindrome $(\LL\arm,\RR\arm)$ is called \intWort{$\alpha$-gapped} 
iff its period~$\qvar$ is at most $\alpha \abs{\LL\arm}$.

Given a gapped repeat $(\LL\arm,\RR\arm)$, it is called \intWort{maximal} 
iff the characters to the immediate left and to the immediate right of its arms differ (as far as they exist), i.e.,
\begin{itemize}
	\item $\warm[\ibeg{\LL\arm}-1] \not= \warm[\ibeg{\RR\arm}-1]$ (or $\ibeg{\LL\arm} = 1$) and 
	\item $\warm[\iend{\LL\arm}+1] \not= \warm[\iend{\RR\arm}+1]$ (or $\iend{\RR\arm} = \abs{\warm}$).
\end{itemize}
Similarly, a gapped palindrome $(\LL\arm,\RR\arm)$ is called \intWort{maximal} 
iff it can be extended neither inwards nor outwards, i.e.,
\begin{itemize}
	\item $\warm[\ibeg{\LL\arm}-1] \not= \warm[\iend{\RR\arm}+1]$ (or $\ibeg{\LL\arm} =1$ or $\iend{\RR\arm}=n$) and
	\item $\warm[\iend{\LL\arm}+1] \not= \warm[\ibeg{\RR\arm}-1]$. 
\end{itemize}
Let $\grGR$ (resp.\ $\gpGR$) denote the set of all maximal $\alpha$-gapped repeats (resp.\ palindromes) in~$\warm$.

Gapped palindromes generalize the definition of ordinary palindromes:
A gapped palindrome $(\LL\arm,\RR\arm)$ is an \intWort{ordinary palindrome} if $\LL\arm \equiv \RR\arm$.
For a maximal gapped palindrome with a gap~$\ibeg{\RR\arm}-\iend{\LL\arm} -1 \le 1$ it follows that $\LL\arm \equiv \RR\arm$ (otherwise it could be extended inwards).

  \begin{figure}[ht]
	\MyFloatBox{%
\begin{tikzpicture}
\pgfmathsetmacro\Mgamma{7/9}
\pgfmathsetmacro\Mnums{6}
\pgfmathsetmacro\MgammaCount{int(floor(\Mnums*\Mgamma))}

\foreach \x in {1,...,\Mnums} {%
    \draw [thin,color=gray!80,dotted] (\x,0) -- (\x,\Mnums);
    \draw (\x,-4pt) -- (\x,4pt)
        node [below,yshift=-2ex] {\x};
}
\foreach \y in {1,...,\Mnums} {%
    \draw [thin,color=gray!80,dotted] (0,\y) -- (\Mnums,\y);
    \draw (-4pt,\y) -- (4pt,\y)
        node [left,xshift=-2ex] {\y};
}
\foreach \y in {1,...,\MgammaCount} {%
    \pgfmathsetmacro\yresult{\y / \Mgamma}
	\draw [color=gray!50] (-3ex,\yresult) -- (4pt,\yresult)
        node [left,color=gray,xshift=-5ex] {\pgfmathprintnumber{\yresult}};
    \draw [thin,color=gray!50] (0,\yresult) -- (\Mnums,\yresult);
}

\foreach \y in {1,3,...,\Mnums} {%
	\pgfmathsetmacro\Mxmod{mod((\y+1)/2,2)+1}
	\pgfmathsetmacro\MxmodP{mod((\y+1)/2,2)+3}
	\foreach \x in {\Mxmod,\MxmodP,...,\Mnums} {%
		\pgfmathsetmacro\mycolorY{mod((\y*7+\x*11*\Mgamma),220)}
		\pgfmathsetmacro\mycolorX{mod((\y*\y*3*\Mnums+\x*31),220)}
		\pgfmathsetmacro\mycolorZ{mod((\y*\y*7*\Mnums+\x*\x*17),220)}
		\definecolor{mycolor}{RGB}{\mycolorX,\mycolorY,\mycolorZ}
		\pgfmathsetmacro\xmin{\x - \y*\Mgamma}
		\pgfmathsetmacro\ymin{\y*(1-\Mgamma)}
		\draw [mycolor,densely dashdotted] (\xmin,\ymin) rectangle (\x,\y);
		\filldraw [mycolor] (\x,\y) circle (2pt);
	}
}

    \draw [<->,thick] (0,\Mnums+1) node (yaxis) [above] {$y$}
        |- (\Mnums+1,0) node (xaxis) [right] {$x$};

\end{tikzpicture}
}{%
	\caption{$7/9$-cover of the points $\lbrace (2x - (y+1 \mod 2), 2y - 1) \mid 1 \le x,y \le 3 \rbrace \subset \N^2$. 
		The dash-dotted rectangle of a point~$\vec{p}$ comprises all points that are $7/9$-covered by~$\vec{p}$ 
		(the rectangle of~$\vec{p}$ is the rectangle that has~$\vec{p}$ as its top right vertex).
		A point~$(x,y)$ with $y=1$ only $7/9$-covers itself.
		The light-gray dotted lines create the grid~$\N^2$. 
		Each value of $i/\gamma$ for $\gamma := 7/9$ and $i \ge 1$ on the $y$-axis is indicated with a gray horizontal line.
	}
  	\label{figPointCover}
}%
  \end{figure}
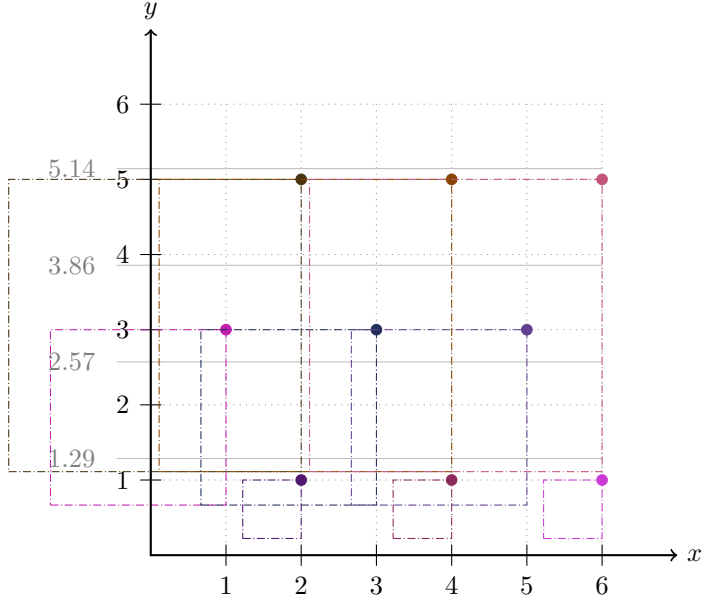
  \section{Improved Point Analysis}\label{secPointAnalysis}
A pair of integers is called a \intWort{point}.
In~\cite{gawrychowski18tighter}, a certain subset~$\SubCover$ of maximal $\alpha$-gapped repeats and maximal $\alpha$-gapped palindromes are mapped to points injectively.
The cardinality of~$\SubCover$ is estimated with the property that 
every point of $\SubCover$ has a large vicinity that does not contain another point of $\SubCover$.
This vicinity is given formally by the following definition:

\begin{definition}
For a real number~$\gamma$ with $\gamma \in (0,1]$, we say that a point $(\hat{x}, \hat{y}) \in \Z^2$ $\gamma$-\intWort{covers} a point $(x, y) \in \Z^2$
iff $\hat{x} - \gamma \hat{y} \le x \le \hat{x}$ and $\hat{y} (1 - \gamma) \le y \le \hat{y}$.
\end{definition}

  	\Cref{figPointCover} gives an example for $\gamma := 7/9$.
In Lemma~7 of~\cite{gawrychowski18tighter}, it is shown that $\abs{\SubCover} < 3n / \gamma$ holds for every set of points $\SubCover \subseteq \Int{1}{n}^2$ with the property that
no two distinct points in~$\SubCover$ $\gamma$-cover the same point.
In the following, we devise an improved version of this lemma to upper bound the number of the $\beta$-aperiodic repeats/palindromes.

\TextFigureBlock{0.8}{%
For our purpose, it is sufficient to focus on the set
$\CoverN{n} := \menge{ (x, y) \mid 1 \le y \le n-1 \text{~and~} 1 \leq x \le n - y }$, 
since we will later show that we can map all maximal $\alpha$-gapped repeats/palindromes to the set
injectively.
Before that, we introduce two small helper lemmas that improve an inequality needed in~\Cref{lemmaPoints}:
}
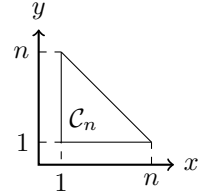
{%
\begin{tikzpicture}[scale=0.6]
    \draw [<->,thick] (0,3) node (yaxis) [above] {$y$}
        |- (3,0) node (xaxis) [right] {$x$};
    \draw (0.5,0.5) coordinate (c) -- (2.5,0.5) coordinate (xmax);
    \draw (c) -- (0.5,2.5) coordinate (ymax);
    \draw (xmax) -- (ymax);

	\node at (1,1) {$\CoverN{n}$};

	\draw[dashed] (yaxis |- ymax) node[left] {$n$} -- (ymax);
	\draw[dashed] (xaxis -| xmax) node[below] {$n$} -- (xmax);

     \draw[dashed] (yaxis |- c) node[left] {$1$}
         -| (xaxis -| c) node[below] {$1$};
\end{tikzpicture}
}%

\begin{lemma}\label{lemmaOneOverGammaIntervalAbs}
	Given a real interval~$\intervalI := [\psi-1/\gamma,\psi)$ with $\gamma,\psi\in \RealNumber$ and $0< \gamma \le 1$, 

	\TextFigureBlock{0.5}{%
	\[
	\abs{\intervalI \cap \Z} =
	\begin{cases}
		\gauss{1/\gamma}+1 & \text{if~} 0 < \psi - \gauss{\psi} \le \delta, \\
		\gauss{1/\gamma} & \text{otherwise},
	\end{cases}
\]
	where $\intervalI \cap \Z = \menge{ i \in \Z \mid i \in \intervalI}$ and $\delta := 1/\gamma - \gauss{1/\gamma}$.
}{%
\begin{tikzpicture}
	\node [color=solarizedViolet] at (2.3,-1.2) {$\gamma = \frac{7}{9}$};
	\node [color=solarizedViolet] at (2.3,-0.8) {$\gamma\psi \in \N$};
    \draw [->] (0,0) -- (4.2,0);
    \foreach \x/\xtext in {0/0,0.5/1,1/2,1.5/3,2/4,2.5/5,3/6,3.5/7,4/8}
      \draw(\x,5pt)--(\x,-5pt) node[below] {\xtext};
	  \foreach \x/\xtext in {0/$\frac{0}{\gamma}$,1.29/$\frac{1}{\gamma}$,2.57/$\frac{2}{\gamma}$,3.85/$\frac{3}{\gamma}$} %
      \draw [color=solarizedViolet,thick] (\x,-5pt)--(\x,5pt) node[above] {\xtext};
	
    \draw[|<->|, yshift=-5ex]  (1,0) -- node[below=0.4ex] {$\delta$}  (1.29,0);
     \draw[-]  (3.5,-0.6) -- node[anchor=south,below=0.2em] {$\gauss{\frac{3}{\gamma}}$}  (3.5,-0.7);
\end{tikzpicture}
}%

\end{lemma}
\begin{proof}
	In the case that $\psi = \gauss{\psi}$ (i.e., $\psi \in \Z$), $\ibeg{\intervalI} = \psi - 1/\gamma \le \psi - \gauss{1/\gamma} \in \intervalI \cap \Z$.
	Hence, $\{\psi - \gauss{1/\gamma}, \ldots, \psi -1\} = \intervalI \cap \Z$, and $\abs{\intervalI \cap \Z} = \gauss{1/\gamma}$.

In the case that $0 < \psi - \gauss{\psi} \le \delta$, we have $\psi - \delta \le \gauss{\psi}$, and therefore
$\ibeg{\intervalI} = \psi - 1/\gamma  = \psi - \gauss{1/\gamma} - \delta \le \gauss{\psi} - \gauss{1/\gamma} \in \intervalI \cap \Z$.
Hence, $\menge{\gauss{\psi} - \gauss{1/\gamma}, \ldots, \gauss{\psi}} = \intervalI \cap \Z$, 
and $\abs{\intervalI \cap \Z} = \gauss{1/\gamma}+1$ (because $\gauss{\psi} < \psi$).

The remaining case is that $\psi - \gauss{\psi} > \delta$.
With $\gauss{\psi} < \psi - \delta = \psi - 1/\gamma + \gauss{1/\gamma}$, we obtain that
$\ibeg{\intervalI} = \psi - 1/\gamma > \gauss{\psi} - \gauss{1/\gamma} \not\in \intervalI \cap \Z$.
Hence, $\menge{\gauss{\psi} - \gauss{1/\gamma}+1, \ldots, \gauss{\psi}} = \intervalI \cap \Z$, and $\abs{\intervalI \cap \Z} = \gauss{1/\gamma}$.
\end{proof}

\begin{lemma}\label{lemmaNonIncreasingFunction}
	Given the function
	$g : \N \rightarrow \N$
	with $g(i) := \abs{\menge{y \in \N \mid (i-1)/\gamma \le y < i/\gamma}}$ for $1 \le i \le \upgauss{n\gamma}$,
	and a nonincreasing function~$f : \N \rightarrow \RealNumber$, the inequality
	\begin{equation}\label{eqGappedNonIncreasingFunction}
		\sum_{i = 1}^{\lceil n \gamma \rceil} (f(i) g(i)) \leq \sum_{i = 1}^{\lceil n \gamma \rceil} f(i) / \gamma
	\end{equation}
	holds for every natural number~$n$ and every real number~$\gamma \in (0,1]$. %
\end{lemma}
\begin{proof}
	We set $Y_i := \menge{y \in \N \mid (i-1)/\gamma \le y < i/\gamma}$.
	Our task is to upper bound the sizes of~$Y_i$, since $g(i) = \abs{Y_i}$.
	It is clear that $\abs{Y_i} \le \gauss{1 / \gamma}+1$. 
	Since~$Y_1$ cannot contain zero, it holds that $\abs{Y_1} \le \gauss{1 / \gamma}$ (if $1/\gamma \in \N$ then $\abs{Y_1} = 1 / \gamma -1$, otherwise $\abs{Y_1} = \gauss{1 / \gamma}$).
	For $i \ge 2$, \Cref{lemmaOneOverGammaIntervalAbs} provides that
	\begin{equation}\label{eqGappedGhighest}
		\abs{Y_i} = \gauss{1/\gamma} + 1 \text{~iff~} 0 < i/\gamma - \gauss{i/\gamma} \le \delta, \text{~where~} \delta := 1 / \gamma - \gauss{1 / \gamma} < 1.
	\end{equation}
Having \cref{eqGappedGhighest}, \cref{eqGappedNonIncreasingFunction} is a conclusion of the following game estimating the cumulative sum of $f(i) / \gamma - f(i) g(i)$:
	The game is divided in $\upgauss{n\gamma}$ rounds. In the $i$-th round ($1 \le i \le \upgauss{n\gamma}$),
    we receive a credit of $(1 / \gamma - \gauss{1 / \gamma}) f(i) = \delta f(i)$,
	but we additionally pay $f(i)$ from the credit when $g(i) = \gauss{1/\gamma} + 1$.
    If the credit does not become negative, 
	it holds that $\sum_{i = 1}^{\lceil n \gamma \rceil} (f(i) g(i)) \leq \sum_{i = 1}^{\lceil n \gamma \rceil} f(i) / \gamma$ (which is what we want to show in this proof).

    Let $i_1, i_2, \dots$ be the sequence of integers such that $g(i_j) = \gauss{1/\gamma} + 1$ for each~$j$.
	After sorting this sequence ascendingly, it holds that $\delta i_j > j$ for every $j$.
	To see this, we write 
	$i/\gamma - \gauss{i/\gamma} = i/\gamma - i \gauss{1/\gamma} - \gauss{i/\gamma - i\gauss{1/\gamma}} = \delta i - \gauss{\delta i}$,
	and apply \cref{eqGappedGhighest}:
	First, $\delta i_1 \ge 1$, since otherwise ($\delta i_1 < 1$) we obtain a contradiction to \cref{eqGappedGhighest} with
	$\delta i_1 - \gauss{\delta i_1} = \delta i_1 > 2 \delta$ (remember that $i_1 \ge 2$ because $\abs{Y_1} \le \gauss{1/\gamma}$).
	Next, assume that there exists a $j \ge 2$ such that $j \le \delta i_j < \delta i_{j+1} < j+1$.
	Then $\delta i_{j+1} - \gauss{\delta i_{j+1}} \ge \delta (i_j + 1) - \gauss{\delta i_j} > \delta$ (since $\delta i_j - \gauss{\delta i_j} > 0$), a contradiction that 
	\cref{eqGappedGhighest} holds for $i_{j+1}$.
	We conclude that $\delta i_j > j$ for every~$j$.
	
	Back to our game, we claim that there is at least $(\delta i_j - j) f(i_j)$ credit remaining after the $i_j$-th round.
	When reaching the $i_1$-th round, we have already gathered a credit of~$\sum_{i=1}^{i_1} \delta f(i)$. 
	Remember that we have to pay the amount $f(i_1)$.
	From our gathered credit we can pay $f(i_1)$ with $s := \delta f(1) + \delta f(2) + \cdots + \delta f(i_1 - 1) + (1 - \delta (i_1 - 1)) f(i_1)$:
	First, $s$ is smaller than our gathered credit, since $f(i_1) < \delta i_1 f(i_1)$, and hence $(1-\delta(i_1-1))f(i_1) < \delta f(i_1)$.
	Second, $s \ge f(i_1)$, because $\delta (i_1-1) f(i_1) \le \sum_{i=1}^{i_1-1} \delta f(i)$ (remember that $f$ is nonincreasing).
	By paying the amount~$s$, a credit of at least $f(i_1) (\delta i_1 - 1)$ remains.

	Under the assumption that our claim holds after the $i_j$-th round for an integer~$j \in \N$, we show that the claim holds after the \mbox{$i_{j+1}$-th} round, too:
	According to our assumption, 
	we have gathered a credit of at least $(\delta i_j - j) f(i_j) + \sum_{i=i_j+1}^{i_{j+1}} \delta f(i)$ at the beginning of the $i_{j+1}$-th round.
	We pay the amount~$f(i_{j+1})$ with
	$s := (\delta i_j - j) f(i_j) + \delta f(i_j + 1) + \cdots + \delta f(i_{j+1} - 1) + (j+1 - \delta (i_{j+1} - 1)) f(i_{j+1})$.
	First, $s$ is smaller than our gathered credit, since $\delta i_{j+1} > j+1$, and hence $(j+1 - \delta (i_{j+1} - 1)) f(i_{j+1}) < \delta f(i_{j+1})$.
	Second, $s \ge f(i_{j+1})$, because
	$\delta (i_{j+1}-1) f(i_{j+1}) \le (\delta i_j - j) f(i_j) + j f(i_{j+1}) + \sum_{i=i_j+1}^{i_{j+1}-1} \delta f(i)$.
	Similar to the $i_1$-th round, a credit of at least~$(\delta i_{j+1} - j - 1) f(i_{j+1})$ remains.
\end{proof}

\begin{lemma}\label{lemmaPoints}
Let $\gamma$ be a real number with $\gamma \in (0,1]$, and
$\SubCover \subseteq \CoverN{n}$ be a set of points such that no two distinct points in~$\SubCover$ $\gamma$-cover the same point.
Then $\abs{\SubCover} < n \pi^2 / (6 \gamma)$. 
In particular, $\abs{\SubCover} \le n \pi^2/6 - 3n/4$ for $\gamma = 1$.
\end{lemma}
\begin{proof}
	Given that a point~$\vec{p}$ in $\Z^2$ is $\gamma$-covered by a point $(\hat{x}, \hat{y})$ of $\SubCover$ with
	$(i-1) / \gamma \leq \hat{y} < i / \gamma$ for a positive integer $i$, we assign~$\vec{p}$ the weight $1/i^2$.
		Otherwise ($\vec{p}$ is not $\gamma$-covered by any point of~$\SubCover$), we assign $\vec{p}$ the weight zero.
	Let us fix a point $(\hat{x}, \hat{y}) \in \SubCover$ with $(i-1) / \gamma \leq \hat{y} < i / \gamma$ for an integer $i$.
	We have
	$\hat{x} - i < \hat{x} - \gamma \hat{y} \le \hat{x} - (i-1)$, and these inequalities also hold when substituting $\hat{x}$ with $\hat{y}$, i.e.,
	$\hat{y} - i < \hat{y} - \gamma \hat{y} \le \hat{y} - (i-1)$.
	There are exactly $i^2$ points~$(x,y) \in \Z^2$ that are $\gamma$-covered by~$(\hat{x},\hat{y})$, since for each of them
	it holds that
	$\hat{x} - i < \hat{x} - \gamma \hat{y} \le \hat{x} - (i-1) \le x \le \hat{x}$ and
	$\hat{y} - i < \hat{y} - \gamma \hat{y} \le \hat{y} - (i-1) \le y \le \hat{y}$.
	Therefore,
	the sum of the weights of the points that are $\gamma$-covered by $(\hat{x}, \hat{y})$ is one.
	As a consequence, the size of $\SubCover$ is equal to the sum of the weights of all points in $\Z^2$.
	In the following, let $\weight[\vec{p}]$ denote the weight of a point~$\vec{p}$.
    In what follows, we upper bound the sum of all weights.

    First, we fix an integer $y$ with $1 \leq y \leq n$, and show that the sum of the weights of all points $(\cdot, y)$ is less than $n / i^2$, where
    $i$ is the integer with $(i-1) / \gamma \leq y < i / \gamma$.
	Given an integer~$x \in \Z$, 
	we conclude by the definition of $\CoverN{n}$ that 
	\[
		\weight[x, y]
	\begin{cases}
	\leq 1/i^2 & \text{for~} 1 \leq x < n - y, \text{and} \\
	= 0 & \text{for~} x \geq n - y.
	\end{cases}
	\]
	The sum $\sum_{x = -\infty}^{1} \weight[x,y]$ is maximized to $1/i$ when each point in 
	$\SetE := \{ (x, y) \in \Z^2 \mid 2 - i \leq x \leq 1 \} $ with $\abs{\SetE} = i$ has weight $1/i^2$, and the other points $\{ (x, y) \in \Z^2 \mid x \le 1 - i  \} $ are not $\gamma$-covered.
	This can be seen by the following fact:
	A point~$(x,y)$ with $x \le 1-i$ can only be $\gamma$-covered by a point~$(\hat{x},\hat{y}) \in \CoverN{n}$ when
	$\hat{x} - \gamma \hat{y} \le x \le 1-i$, or equivalently $i \le \gamma \hat{y}$ (the smallest value for~$\hat{x}$ is one).
	Assume that such a point~$(\hat{x},\hat{y})$ exists. 
	Then there is an integer $j$ with $i < j$ such that $\gamma \hat{y} < j$ and $(j-1)/\gamma \le \hat{y} < j/\gamma$. 
	Since $1-j \le \hat{x} - \gamma \hat{y} \le x \le 1$, there are at most $\abs{\menge{(x, y) \mid 2 - j \leq x \leq 1 }} = j$ many different values for~$x$.
	Furthermore, since $(\hat{x},\hat{y}) \in \CoverN{n}$ $\gamma$-covers $(x,y)$, it is not possible that another element of $\CoverN{n}$
	$\gamma$-covers~$(x',y)$ with $x' < x$ (otherwise it would also cover~$(x,y)$).
    In total, the sum under consideration~$\sum_{x \le 1} \weight[x,y]$ can be at most~$1/j$, which is less than~$1/i$.
	With $\sum_{x \leq 1} \weight[x,y] \le 1/i$ we obtain
 $\sum_{x \in \Z} \weight[x,y] \le (n - y - 1 + i)/i^2 \le (n - y + \gamma y)/i^2 \le n/i^2$.

	Having computed~$\sum_{x \in \Z} \weight[x,y]$ for a fixed~$y$, we compute the sum over all~$y$ with $y\in\Z$.
	First, we deal with the special case that $\gamma = 1$. 
	That is because it is the only case where $\weight[\cdot, 0]$ might not be zero (given $(\hat{x},\hat{y}) \in \CoverN{n}$ and $\gamma < 1$, it holds that $\hat{y} \ge 1$ and therefore $0 < \hat{y} - \gamma \hat{y}$).
	A point~$(x,y)$ is $1$-covered by~$(\hat{x},\hat{y}) \in \CoverN{n}$ iff $0 \le y \le \hat{y}$ and $\hat{x}- \hat{y} \le x \le \hat{y}$ hold.
	The weight of a point~$(x,0)$ with $0 \leq x \leq n - 1$ is maximized to $1/2^2$ if it is $\gamma$-covered by a point~$(\hat{x},\hat{y}) \in \CoverN{n}$ with the lowest possible value of~$\hat{y}$, which is one.
	We conclude that $\sum_{x\in \Z} \weight[x, 0] \leq n/2^2$.
	With the same argument we conclude that $\sum_{x \in \Z} \weight[x,y] \le n / (y+1)^2$ for every positive integer $y$.
	Summing up everything yields
	$\sum_{(x, y) \in \Z^2} \weight[x,y] \le n/2^2 + n \sum_{y = 1}^{n}(1/(y+1)^2) = n/4 + n \sum_{i=2}^{\infty} (1/i^2) = n/4 + n \pi^2/6 - n = n \pi^2/6 - 3n/4$ due to the Basel problem.

    Finally we consider the case that $\gamma < 1$.
	The idea is to cover the interval $\Int{1}{n-1}$ with the sets $Y_i := \menge{y \in \N \mid (i-1)/\gamma \le y < i/\gamma}$ for $1 \le i \le \upgauss{n\gamma}$.
	Since a point $(x,y_i)$ with $y_i \in Y_i$ has a weight of at most $1/i^2$, summing up all weights gives
	$\sum_{(x,y)\in\Z^2} \weight[x, y] \le \sum_{i=1}^{\upgauss{n\gamma}} n \abs{Y_i}/i^2$.
	To compute~$\abs{Y_i}$, we use the function $g(i) := \abs{Y_i}$ as defined in \Cref{lemmaNonIncreasingFunction}.
    With $g$ the upper bound of $\sum_{(x, y) \in \Z^2} \weight[x,y]$ can be stated as $\sum_{i = 1}^{\lceil n \gamma \rceil} (g(i) n / i^2)$.
    Since $g(i) \le \gauss{1 / \gamma} + 1$, it is easy to see that
    $\sum_{i = 1}^{\lceil n \gamma \rceil} (g(i) n / i^2) < n (\gauss{1 / \gamma} + 1) \sum_{i = 1}^{\lceil n \gamma \rceil} (1/i^2) < n (\gauss{1 / \gamma} + 1) \pi^2 / 6$.
	By defining the non-increasing function~$f$ with $f(i) := n/i^2$,
	\Cref{lemmaNonIncreasingFunction} yields 
	$\sum_{i = 1}^{\upgauss{n\gamma}} (g(i) n / i^2) = \sum_{i=1}^{\upgauss{n\gamma}} g(i) f(i) \le 
	(n / \gamma) \sum_{i = 1}^{\upgauss{n \gamma}} (1 / i^2) < \sum_{i = 1}^{\infty} n / (\gamma i^2) = n \pi^2 / (6 \gamma)$,
    which is also an upper bound of $\abs{\SubCover}$.
\end{proof}

By restricting the subset $\SubCover \subseteq \CoverN{n}$ in \Cref{lemmaPoints} to be additionally bijective to the set of all maximal $\alpha$-gapped repeats or palindromes,
we can refine the upper bound attained in \Cref{lemmaPoints}.
For the maximal $\alpha$-gapped repeats, we follow the approach of \citet{gawrychowski18tighter} 
who map a maximal $\alpha$-gapped repeat $(\LL\arm,\RR\arm)$ with period~$\qvar := \ibeg{\RR\arm}-\ibeg{\LL\arm}$ to $(\iend{\LL\arm}, \qvar)$.
It holds that $(\iend{\LL\arm}, \qvar) \in \CoverN{n}$, because
$\iend{\RR\arm}$ and $\qvar$ are positive, and $\iend{\LL\arm}+\qvar = \iend{\RR\arm} \le n$.
In particular $\iend{\LL\arm} \le n-1$,
since otherwise ($\iend{\LL\arm} = n$) both endings $\iend{\RR\arm}$ and $\iend{\LL\arm}$ would be equal, and therefore $\LL\arm \equiv \RR\arm$ (a contradiction to the definition of gapped repeats).
Let~$\mapRepeat$ denote this mapping, and let
$\imgRepeat := \menge{\mapRepeat(\LL\arm,\RR\arm) \mid \agr \text{~is a maximal~} \alpha\text{-gapped repeat} } \subset \CoverN{n}$ denote the image of $\mapRepeat$.
The following \lcnamecref{lemmaPointsRepPropA} bounds the size of $\imgRepeat$ to be roughly at half of the size of~$\CoverN{n}$, a fact that will be used in \Cref{lemmaPointsRep}.
\begin{lemma}\label{lemmaPointsRepPropA}
	If $(x, y) \in \imgRepeat$, then $(x+1,y) \notin \imgRepeat$.
\end{lemma}
\begin{proof}
	Let $(\LL\arm,\RR\arm)$ be a maximal $\alpha$-gapped $\beta$-aperiodic repeat 
	with period~$\qvar = \ibeg{\RR\arm}-\ibeg{\LL\arm}$, and $(x,y) := \mapRepeat(\LL\arm,\RR\arm) = (\iend{\LL\arm}, \qvar)$.
	If $(x + 1, y) \in \imgRepeat$, then $\warm[x + 1] = \warm[\iend{\LL\arm} + 1] = \warm[x + y + 1] = \warm[\iend{\RR\arm} + 1]$, which contradicts the maximality of $(\LL\arm,\RR\arm)$.
\end{proof}

With \Cref{lemmaPointsRepPropA} we attain a version of \Cref{lemmaPoints} tailored to subsets of~$\imgRepeat$:

\begin{figure}[ht]
		\usetikzlibrary{calc}
	\MyFloatBox{%
		\begin{tikzpicture}
		\pgfmathsetmacro\Mgamma{7/9}
		\pgfmathsetmacro\MnumsX{3}
		\pgfmathsetmacro\MnumsY{2}
		\pgfmathsetmacro\MgammaCount{int(floor(\MnumsY*\Mgamma))}

		\filldraw [] (1,3) circle (2pt) node [anchor=west] {$\in \SetE \setminus \SubCover$};
		\filldraw [color=solarizedViolet] (1,2.5) circle (2pt) node [anchor=west] {$\in \SubCover$};
		\draw [color=gray] (0.8,2.3) rectangle (2.5,3.3);
		
		\foreach \x in {1,...,\MnumsX} {%
		    \draw [thin,color=gray!80,dotted] (\x,0) -- (\x,\MnumsY);
		}
		    \draw (1,-4pt) -- (1,4pt) node [below,yshift=-2ex] {$x$};
		    \draw (2,-4pt) -- (2,4pt) node [below,yshift=-2ex] {$x+1$};
		    \draw (3,-4pt) -- (3,4pt) node [below,yshift=-2ex] {$x+2$};
		\foreach \y in {1,...,\MnumsY} {%
		    \draw [thin,color=gray!80,dotted] (0,\y) -- (\MnumsX,\y);
		}
		    \draw (-4pt,1) -- (4pt,1) node [left,xshift=-2ex] {$y$};
			\draw (-4pt,2) -- (4pt,2) node [left,xshift=-2ex] {$\hat{y}$};
		\foreach \y in {1,...,\MgammaCount} {%
		    \pgfmathsetmacro\yresult{\y / \Mgamma}
			\draw [color=gray!50] (-3ex,\yresult) -- (4pt,\yresult)
			node [left,color=gray,xshift=-5ex] {$1/\gamma$};
		    \draw [thin,color=gray!50] (0,\yresult) -- (\MnumsX+1,\yresult);
		}
		
		\filldraw [color=solarizedViolet] (1,1) circle (2pt);
		\draw [color=solarizedViolet,densely dashdotted] ($(1,1)-\Mgamma*(1,1)$) rectangle (1,1);
		\filldraw [] (2,1) circle (2pt);
		\filldraw [] (3,1) circle (2pt);
		\filldraw [color=solarizedViolet] (3,2) circle (2pt);
		\draw [color=solarizedViolet,densely dashdotted] ($(3,2)-\Mgamma*(2,2)$) rectangle (3,2);
		
		    \draw [<->,thick] (0,\MnumsY+1) node (yaxis) [above] {}
		        |- (\MnumsX+1,0) node (xaxis) [right] {};
		
		\end{tikzpicture}
		}{%
		\caption{Setting of the proof of \Cref{lemmaPointsRep}, where the point~$(x,y) \in \SubCover$, but $(x+1,y) \not\in \SubCover$ with $\weight[x+1,y] > 0$. 
			Thus $(x+1,y)$ is $\gamma$-covered by a point~$(\hat{x},\hat{y}) \in \SubCover$ ($\hat{x} = x+2$ in this figure). 
			Like in \Cref{figPointCover}, the dash-dotted rectangle of a point~$\vec{p} \in \SubCover$ comprises all points that are $\gamma$-covered by~$\vec{p}$.
			The points that are $\gamma$-covered by~$(\hat{x},\hat{y})$ are contained in the top right dashed rectangle.
			It can be seen that $(x+2,y)$ is also $\gamma$-covered by~$(\hat{x},\hat{y})$, and therefore cannot be in~$\SubCover$.
		}
	\label{figPointsRep}
}%
\end{figure}
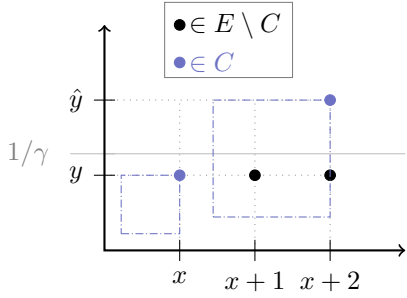

\begin{lemma}\label{lemmaPointsRep}
    Let $\gamma$ be a real number with $\gamma \in (0,1]$.
	A set of points $\SubCover \subseteq \imgRepeat$ such that no two distinct points in~$\SubCover$ $\gamma$-cover the same point obeys the inequality 
	$\abs{\SubCover} < n (\pi^2 / 6 - 1/2)/\gamma$.
\end{lemma}
\begin{proof}
    If $\gamma = 1$, \Cref{lemmaPoints} already gives $\abs{\SubCover} < n \pi^2/6 - 3n/4 < n \pi^2 / 6 - n/2$.
	For the case $\gamma < 1$,
	we focus on the points~$\SetE := \menge{(x,y) \mid 1 \le x \le n \text{~and~} y < 1 / \gamma}$.
	In the proof of \Cref{lemmaPoints}, we used the weights~$\weight[\cdot]$ of all points in~$\Z^2$ as an upper of~$\abs{\SubCover}$.
	There, we bounded the sum of the weights of all points in~$\SetE$ by~$n/\gamma$ (assign each point the weight~$1$).
	We can refine this upper bound by halving the weights of the points in~$\SetE$.
	We justify this with the following analysis.

	First, each point $(n,y) \in \SetE$ has weight zero, since there is no point $(\hat{x},\hat{y}) \in \SubCover$ (and even in $\CoverN{n}$)
	with $n \le \hat{x}$. Thus the sum of the weights of the points $(n-1,y)$ and $(n,y)$ is at most one, for every~$y$ with $1 \le y < 1/\gamma$.

	Second, a point $(x,y) \in \SetE \cap \SubCover$ can only cover itself, since $y < 1/\gamma$.
	Consequently, a point $(x,y) \in \SetE \setminus \SubCover$ can have a weight of at most $1/2^2 = 1/4$, 
	since all points~$(\hat{x},\hat{y}) \in \SubCover \setminus \SetE$ have $\hat{y} \ge 1/\gamma$.
	Given that $\SetE \cap \SubCover = \emptyset$, the total weight of all points in $\SetE$ is at most $(1/4) \abs{\SetE}$.

	Finally, suppose there is a point $(x,y) \in \SetE \cap \SubCover$. 
	Then $\weight[x,y] = 1$.
	Given $x \le n-2$, $(x+1,y) \in \CoverN{n}$, but $(x+1,y) \not\in \SubCover$ according to \Cref{lemmaPointsRepPropA}.
	We consider two cases:

		\begin{itemize}
			\item $\weight[x+1,y] = 0$. Then both points $(x,y)$ and $(x+1,y)$ together have a weight of one.
			\item $\weight[x+1,y] > 0$, see also \Cref{figPointsRep}. Since $(x+1,y) \not\in \SubCover$, $\weight[x+1,y] \le 1/4$, i.e., it is $\gamma$-covered by a point $(\hat{x},\hat{y}) \in \SubCover \setminus \SetE$.
				Since $\hat{y} \ge 1/\gamma$, the point~$(\hat{x},\hat{y})$ $\gamma$-covers at least four points (including itself).
				Since $\weight[x,y] =1$, $(\hat{x},\hat{y})$ cannot $\gamma$-cover $(x,y)$. 
				Instead, it $\gamma$-covers the point~$(x+2,y)$.
				We conclude that $(x+2,y) \not\in \SubCover$.
				All three points $(x,y)$, $(x+1,y)$, and $(x+2,y)$ have a total weight of at most $1 + 1/4 + 1/4 = 3/2$.
		\end{itemize}

	In both cases, a node has the average weight of at most~$1/2$.
	Summing up all average weights yields the total weight of all points in~$\SetE$, which is at most $(1/2)\abs{\SetE} = n/(2\gamma)$.

	Following the proof of \Cref{lemmaPoints}, our modification of the weights modifies the 
	nonincreasing function~$f$, which is now defined by $f(1) := n/2$ and $f(i) := n/i^2$ for $i \ge 2$.
Modifying $f$ yields the upper bound
    $\sum_{i = 1}^{\lceil n \gamma \rceil} (f(i) g(i)) \leq n (1/2 + \sum_{i = 2}^{\lceil n \gamma \rceil} (1 / i^2))/\gamma < n (\pi^2 / 6 - 1/2)/\gamma$
	on the size of~$\SubCover$.
\end{proof}

This result already improves the upper bound on the maximum number of all maximal $\alpha$-gapped repeats.
The improvement is clarified in the following theorem:

\begin{theorem}\label{thmMaxReps}
    Given a real number $\alpha$ with $\alpha > 1$ and a word~$\warm$ of length~$n$, 
	the number of all $\alpha$-gapped repeats $\abs{\grGR}$ is less than $\RepeatResult$.
\end{theorem}
\begin{proof}
	We follow the approach of \cite[Theorem~11]{gawrychowski18tighter}, 
	where $\grGR$ is split into a set of $\beta$-\emph{periodic} maximal $\alpha$-gapped repeats \grP{} and $\beta$-\emph{aperiodic} maximal $\alpha$-gapped repeats \grNP{}, for a real number~$\beta$ with $2/3 \le \beta < 1$.
	The set~\grP{} has at most $2 \alpha \sumExp{\warm} /\beta$ elements due to \cite[Lemma~8]{gawrychowski18tighter}.
	Combining the results of \Cref{lemmaPointsRep} and \cite[Lemma~9]{gawrychowski18tighter} yields that the set~\grNP{}
	has at most $(\pi^2/6 - 1/2) \alpha n / (1 - \beta)$ elements.
	Summing up the sizes of both sets yields
	$\abs{\grGR} < 2 \alpha \sumExp{\warm} /\beta + (\pi^2/6 - 1/2) \alpha n / (1 - \beta)$.
	This number becomes minimal with
	$\abs{\grGR} < 9 \alpha n + 3(\pi^2/6 - 1/2) \alpha n = \RepeatResult$
	when setting $\beta$ to $2/3$. 
\end{proof} %

\section{On the Number of all Maximal \texorpdfstring{$\alpha$}{Alpha}-gapped Palindromes}
Our approach is to partition the set of all maximal $\alpha$-gapped palindromes~$\gpGR$ into subsets,
and to analyze these subsets individually, whose definitions follow:
Given a real number~$\beta > 0$, 
a gapped palindrome $(\LL\arm,\RR\arm)$ with $\LL\arm \not\equiv \RR\arm$ belongs to 
the set of all maximal $\alpha$-gapped $\beta$-periodic palindromes~\gpP{}
iff $\LL\arm$ contains a periodic suffix of length at least $\beta \abs{\LL\arm}$.
We call the elements of $\gpP{}$ \intWort{$\beta$-periodic}.
If a maximal $\alpha$-gapped palindrome~$(\LL\arm,\RR\arm)$ is neither $\beta$-periodic nor a maximal ordinary palindrome,
we call it \intWort{$\beta$-aperiodic}. 
The set of all maximal $\alpha$-gapped $\beta$-aperiodic palindromes is denoted by~$\gpNP$.
To sum up, we partition the set of all maximal $\alpha$-gapped palindromes~\gpGR{} in 
\begin{itemize}
	\item the set of all maximal $\alpha$-gapped $\beta$-periodic palindromes $\gpP$,
	\item the set of all maximal ordinary palindromes, and
	\item the set of all maximal $\alpha$-gapped $\beta$-aperiodic palindromes $\gpNP$.
\end{itemize}
The size of the second set is known to be at most $2\abs{\warm}-1$.
In the following, we give an upper bound on the number of maximal $\alpha$-gapped palindromes that are $\beta$-periodic or $\beta$-aperiodic with \Cref{lemmaP} or \Cref{lemmaPointsPali}, respectively.

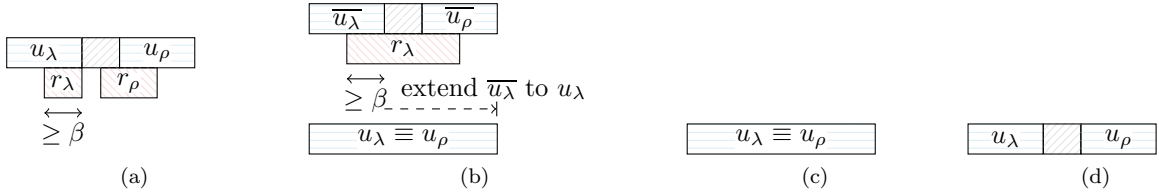
\begin{figure}[ht]
	\centering{%
\hfill\subfloat[ ] {%
			\begin{tikzpicture}[yscale=0.4,xscale=0.5]
			\col[arm]{2}{$\LL\arm$}
			\col[gap]{1}{}
			\col[arm]{2}{$\RR\arm$}
			\newrow
			\void{1}
			\col[rep]{1}{$\LL\rep$}
			\void{0.5}
			\col[rep]{1.5}{$\RR\rep$}
			\newrow
			\newrow
			\void{1}
			\dist[$\ge \beta$]{1}
			\end{tikzpicture}
			\hspace{2em}
			}\hfill\subfloat[\label{figPaliTypesB} ]{%
			\begin{tikzpicture}[yscale=0.4,xscale=0.5]
			\col[arm]{2}{$\s{\LL\arm}$}
			\col[gap]{1}{}
			\col[arm]{2}{$\s{\RR\arm}$}
			\newrow
			\void{1}
			\col[rep]{3}{$\LL\rep$}
			\newrow
			\newrow
			\void{1}
			\dist[$\ge \beta$]{1}
	\draw [dashed,->|] (\ourCol,\ourRow-0.5) -- (\ourCol+3,\ourRow-0.5) node [midway,anchor=south] {\hspace{4em}extend~$\s{\LL\arm}$ to $\LL\arm$};
			\newrow
			\col[arm]{5}{$\LL\arm \equiv \RR\arm$}
			\end{tikzpicture}
			\hspace{2em}
}\hfill\subfloat[ ] {%
			\begin{tikzpicture}[yscale=0.4,xscale=0.5]
			\col[arm]{5}{$\LL\arm \equiv \RR\arm$}
			\newrow
			\void{1}
			\void{1}
			\newrow
			\void{1}
			\end{tikzpicture}
			\hspace{2em}
}\hfill\subfloat[ ] {%
\begin{tikzpicture}[yscale=0.4,xscale=0.5]
\col[arm]{2}{$\LL\arm$}
\col[gap]{1}{}
\col[arm]{2}{$\RR\arm$}
\end{tikzpicture}
			\hspace{2em}
}%

	}%
	\caption{%
	Types of maximal $\alpha$-gapped palindromes~$(\LL\arm,\RR\arm)$ under consideration.
	$\LL\rep$ and $\RR\rep$ are runs.
	Figure~(a) shows a $\beta$-periodic $\alpha$-gapped palindrome, counted in \Cref{lemmaP}.
	The run $\LL\rep$ in Figure~(b) covering a suffix of length~$\beta$ of the left arm~$\LL\arm$ of a maximal $\beta$-periodic gapped palindrome enforces that $\ibeg{\RR\arm} \le 2 + \iend{\LL\arm}$ (see proof of \Cref{lemmaP}), i.e., Figure~(b) shows that the gapped palindrome~$(\s{\LL\arm},\s{\RR\arm})$ can be extended inwards to form a maximal ordinary palindrome.
	Figure~(c) shows an even palindrome,
	Figure~(d) shows an $\alpha$-gapped $\beta$-aperiodic palindrome, counted in \Cref{lemmaPalNP}.
  }
	\label{figPaliTypes}
\end{figure}

\begin{lemma}\label{lemmaP}
Let $\warm$ be a word, and $\alpha$ and $\beta$ two real numbers with $\alpha > 1$ and $0 < \beta < 1$.
Then $\abs{\gpP}$ is at most $2 (\alpha-1) \sumExp{\warm} / \beta$. 
\end{lemma}
\begin{proof}
  Let $\agr \in \gpP$. 
  By definition, the left arm~$\LL\arm$ has a periodic suffix~$\LL\subarm$ of length at least $\beta \abs{\LL\arm}$.
  Let $\LL\rep$ denote the run that generates $\LL\subarm$, i.e.,
  $\LL\subarm \subseteq \LL\rep$. 
  By the definition of the gapped palindromes, there is a reverse copy $\RR\subarm$ of $\LL\subarm$ contained in $\RR\arm$ with
  $\RR\subarm \equiv \substr{\warm}{\ibeg{\RR\arm}}{ \ibeg{\RR\arm}+\abs{\LL\subarm}-1}$ and $\RR\subarm = \pali{\LL\subarm}$.
  Let $\RR\rep$ be the run generating $\RR\subarm$.
  By definition, $\RR\rep$ has the same period~$\period$ as $\LL\rep$.

  If $\LL\rep \equiv \RR\rep$ (see \Cref{figPaliTypesB}), then either~$\ibeg{\RR\arm}-\iend{\LL\arm} \le 2$ (i.e., $\agr$ is an ordinary palindrome), 
  or $\agr$ is not maximal.
  That is because of the following: 
  Assume that $\LL\rep$ contains~$\LL\subarm$ and $\RR\subarm$.
  Then we have
  $\warm[\iend{\LL\subarm} + 1] = \warm[\iend{\LL\subarm}-\period+1] = \warm[\ibeg{\RR\subarm}+\period-1] = \warm[\ibeg{\RR\subarm}-1]$, 
  where the first and third equality follows from $\abs{\RR\subarm} = \abs{\LL\subarm} \ge 2\period$,
  and the second equality follows from $\RR\subarm = \pali{\LL\subarm}$.
  
  From now on, we assume that $\LL\rep \not\equiv \RR\rep$.
  Since $\agr$~is maximal, $\iend{\LL\arm} = \iend{\LL\rep}$ or $\ibeg{\RR\arm} = \ibeg{\RR\rep}$ must hold;
  otherwise we could extend $\agr$ inwards.
  This means that~$\agr$ is uniquely determined by the gap $\gap := \ibeg{\RR\arm}-\iend{\LL\arm} -1 $ and
	  \begin{enumerate}[(a)]
	\item $\LL\rep$ in case $\iend{\LL\arm} = \iend{\LL\rep}$, or \label{itPeriodicAPal}
	\item $\RR\rep$ in case $\ibeg{\RR\arm} = \ibeg{\RR\rep}$. \label{itPeriodicBPal}
	  \end{enumerate}
	  Since ordinary palindromes are excluded from the set of all maximal $\alpha$-gapped $\beta$-periodic palindromes,
	  the gap~$\gap$ is at least two.
	  Cases~\ref{itPeriodicAPal} and~\ref{itPeriodicBPal} are depicted in \Cref{figPeriodicPal}.

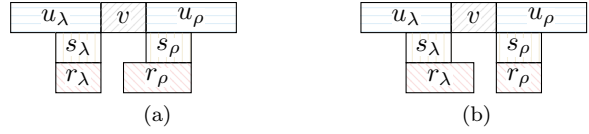
\begin{figure}[ht]
	{%
\subfloat[] {%
			\begin{tikzpicture}[yscale=0.4,xscale=0.6]
			\col[arm]{2}{$\LL\arm$}
			\col[gap]{1}{$\gap$}
			\col[arm]{2}{$\RR\arm$}
			\newrow
			\void{1}
			\col[subarm]{1}{$\LL\subarm$}
			\void{1}
			\col[subarm]{1}{$\RR\subarm$}
			\newrow
			\void{1}
			\col[rep]{1}{$\LL\rep$}
			\void{0.5}
			\col[rep]{1.5}{$\RR\rep$}
			\end{tikzpicture}
			\hspace{2em}
		}\hspace{2em}\subfloat[ ] {%
\begin{tikzpicture}[yscale=0.4,xscale=0.6]
\col[arm]{2}{$\LL\arm$}
\col[gap]{1}{$\gap$}
\col[arm]{2}{$\RR\arm$}
\newrow
\void{1}
\col[subarm]{1}{$\LL\subarm$}
\void{1}
\col[subarm]{1}{$\RR\subarm$}
\newrow
\void{1}
\col[rep]{1.5}{$\LL\rep$}
\void{0.5}
\col[rep]{1}{$\RR\rep$}
\end{tikzpicture}
}%
}{%
	\caption{%
		Setting of the proof of \Cref{lemmaP}.
		Each figure depicts a maximal $\alpha$-gapped $\beta$-periodic palindrome~$(\LL\arm,\RR\arm)$ with the periodic suffix $\LL\subarm$.
		The periodic suffix
		$\LL\subarm \equiv \LL\rep \cap \LL\arm$ of $\LL\arm$ and the periodic prefix $\RR\subarm \equiv \RR\rep \cap \RR\arm$ of $\RR\arm$ 
		are the intersections of the runs $\LL\rep$ and $\RR\rep$ with the respective arms.
		By the maximality property of runs, 
		the equation~\ref{itPeriodicAPal}~$\iend{\LL\arm} = \iend{\LL\rep}$ or~\ref{itPeriodicBPal}~$\ibeg{\RR\arm} = \ibeg{\RR\rep}$ must hold.
  }
	\label{figPeriodicPal}
}%
\end{figure}

  We analyze Case~\ref{itPeriodicAPal} with $\iend{\LL\subarm} = \iend{\LL\rep}$, 
  Case~\ref{itPeriodicBPal} is treated exactly in the same way by symmetry.
    The gapped palindrome~$\agr$ is identified by its gap~$\gap \ge 2$ and $\LL\rep$.
    We fix $\LL\rep$ and count the number of possible values of~$\gap$.
	Since the starting position $\ibeg{\RR\subarm} = \iend{\LL\rep} + \gap +1$ of the periodic segment~$\RR\subarm$ is determined by~$\gap$, two possible values of~$\gap$ must have a distance of at least~$\period$ due to \Cref{lemmaRepetitiveOccs}.
    Since $\abs{\LL\arm} \le \abs{\LL\subarm} / \beta$ and $\agr$ is $\alpha$-gapped,
    $\gap \leq (\alpha - 1) \abs{\LL\arm} \leq (\alpha - 1) \abs{\LL\subarm} / \beta$.
    Then the number of possible values for $\gap$ is bounded by 
    $\abs{\LL\subarm} (\alpha-1) / (\beta \period) = \abs{\LL\rep} (\alpha-1) / (\beta \period) = \exp(\LL\rep) (\alpha-1) / \beta$.
	In total, the number of maximal $\alpha$-gapped palindromes in this case is bounded by $(\alpha-1) \sumExp{\warm} / \beta$ for the case
	$\iend{\LL\arm} = \iend{\LL\rep}$.
	Case~\ref{itPeriodicBPal} is symmetric, leading to the bound of $2 (\alpha-1) \sumExp{\warm} / \beta$ in total.
  \end{proof}

To apply the results of \Cref{secPointAnalysis}, we map maximal $\alpha$-gapped $\beta$-aperiodic palindromes to points.
\citet{gawrychowski18tighter} map 
a maximal $\alpha$-gapped $\beta$-aperiodic palindrome~$(\LL\arm,\RR\arm)$ to the point $\tuple{\iend{\LL\arm}, \gap}$,
where $\gap := \ibeg{\RR\arm}-\iend{\LL\arm}-1$ is the gap of~$(\LL\arm, \RR\arm)$.
Since~$(\LL\arm,\RR\arm)$ is $\beta$-aperiodic, the gap $\gap$ is at least two (otherwise it could be extended inwards to a maximal ordinary palindrome).
With $\iend{\LL\arm} + \gap = \ibeg{\RR\arm} -1 \le n-1$, we conclude that $\tuple{\iend{\LL\arm}, \gap} \in \CoverN{n}$.
However, this mapping seems not useful in combination with \emph{our} definition of the $\beta$-periodic gapped palindromes.
Defining periodic gapped palindromes to have a left arm with a periodic \emph{suffix} (instead of prefix as in~\cite{gawrychowski18tighter})
invalidates the proof of Lemma~12 in~\cite{gawrychowski18tighter}.
There, we fail to transfer the contradiction in the Sub-Case~2b with $2\zvar - \delta < 0$ to our new definition:
We want to derive a contradiction by showing that $\LL\arm$ has a sufficiently large periodic suffix~$\LL\subarm$
(in \cite[Lemma 12]{gawrychowski18tighter}, it was shown that $\LL\arm$ has a sufficiently large periodic prefix).
However, we have not found a way to
upper bound the length of $\LL\arm$, and thus, we are not able to show that the periodic suffix $\LL\subarm$ is sufficiently large in relation to $\abs{\LL\arm}$.

To solve this problem, we define an alternative mapping~$\mapPali$ that 
maps a maximal $\alpha$-gapped $\beta$-aperiodic palindrome~$(\LL\arm,\RR\arm)$ of a word of length~$n$ to the point
\[
(m, d) := \mapPali(\LL\arm,\RR\arm) := (\upgauss{(\ibeg{\LL\arm}+\iend{\LL\arm})/2 }, \gauss{(\ibeg{\RR\arm}+\iend{\RR\arm})/2} - \upgauss{(\ibeg{\LL\arm}+\iend{\LL\arm})/2 }).
\]
Let $\imgPali = \menge{\mapPali(\LL\arm,\RR\arm) \mid \agr \text{~is a maximal~} \alpha\text{-gapped~} \beta\text{-aperiodic palindrome} } \subset \CoverN{n}$ be the image of~$\mapPali$.
The first coordinate~$m$ is the (integer) position nearest to the mid-point $(\ibeg{\LL\arm}+\iend{\LL\arm})/2$ (tie-breaking to the \emph{right}) of the left arm,
and $m + d$ is the position nearest to the mid-point $(\ibeg{\RR\arm}+\iend{\RR\arm})/2$ (tie-breaking to the \emph{left}) of the right arm (in particular, $\warm[m] = \warm[m+d]$).
The mapping~$\mapPali$ is injective because we can retrieve the pair of segments~$(\LL\arm,\RR\arm)$ by computing the maximal inward and outward matches at the positions $m$ and $m + d$.
Since $m$ and $d$ are positive integers with $m+d \le n$, we conclude that $(m, d) \in \CoverN{n}$. 
For convenience, we give an alternative definition of $(m, d)$ using the function $\corr{i} := (i+1 \mod 2)/2 $ such that $\corr{i} = 0$ if $i$ is odd, and $\corr{i} = 1/2$ if $i$ is even.
With $\corr{\ibeg{\LL\arm}+\iend{\LL\arm}+1} = \corr{2\ibeg{\LL\arm} + \abs{\LL\arm}} = \corr{\abs{\LL\arm}}$ we get
$(m, d) = ((\ibeg{\LL\arm}+\iend{\LL\arm})/2 + \corr{\abs{\LL\arm}}, \qvar - 2 \corr{\abs{\LL\arm}})$,
where $\qvar := \ibeg{\RR\arm} - \ibeg{\LL\arm}$ is the period of~$\agr$ (see also \Cref{figGapPoint}).

  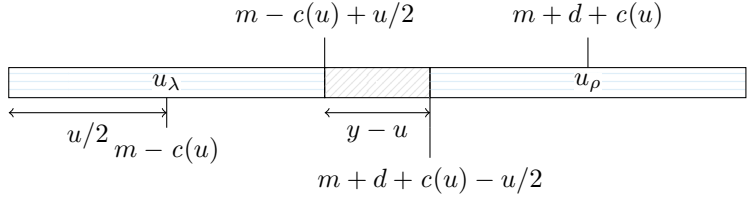
\begin{figure}[t]
    \MyFloatBox{%
		\begin{tikzpicture}[yscale=0.4,xscale=0.7]
 \void{3}
  \void{3}
  \markeup{$m - \corr{\Sarm} + \Sarm/2$}
  \void{5}
  \markeup{$m+d+\corr{\Sarm}$}
  \newrow
  \col[arm]{6}{$\LL\arm$}
  \col[gap]{2}{}
  \col[arm]{6}{$\RR\arm$}
  \newrow
  \void{3}
  \marke{$m - \corr{\Sarm}$}
  \void{5}
  \marke{}
  \newrow
  \dist[$\Sarm/2$]{3}
  \void{3}
  \dist[$y-\Sarm$]{2}
  \marke{$m+d+\corr{\Sarm}-\Sarm/2$}
	\end{tikzpicture}
		}{%
		\caption{A gapped palindrome $(\LL\arm,\RR\arm)$ with $\Sarm = \abs{\LL\arm}$ mapped to the point~$(m,d)$.}
\label{figGapPoint}
		}%
  \end{figure}

  \begin{fact}\label{corGappedCorrectionPlaces}
	  Given a maximal gapped palindrome~$(\LL\arm,\RR\arm)$ with $\Sarm := \abs{\LL\arm}$, it holds that
	  \begin{enumerate}[(a)]
		  \item $\abs{\Sarm/2 - \corr{\Sarm}} \in \menge{1/2,3/2,5/2,\ldots}$, \label{itGappedCorrectionPlacesDiff}
		  \item $\ibeg{\LL\arm} = m - \corr{\Sarm} - \Sarm/2 + 1/2$, \label{itGappedCorrectionPlacesBLArm}
		  \item $\iend{\LL\arm} = m - \corr{\Sarm} + \Sarm/2 - 1/2$, and \label{itGappedCorrectionPlacesELarm}
		  \item $\ibeg{\RR\arm} = m + d + \corr{\Sarm} - \Sarm/2 + 1/2$. \label{itGappedCorrectionBRArm}
		  \item If $d := \gauss{(\ibeg{\RR\arm}+\iend{\RR\arm})/2} - \upgauss{(\ibeg{\LL\arm}+\iend{\LL\arm})/2 } \le 2$, then $(\LL\arm,\RR\arm)$ is a maximal \emph{ordinary} palindrome. \label{itGappedCorrectionOrdinary}
	  \end{enumerate}
  \end{fact}

\begin{lemma}\label{lemmaPointsPalProp}
Given a maximal $\alpha$-gapped $\beta$-aperiodic palindrome $(\LL\arm, \RR\arm)$ with $\Sarm := \abs{\LL\arm}$ and $(m, d) = \mapPali(\LL\arm, \RR\arm)$,
$(m + i, d - 2i) \notin \imgPali$ for every integer $i$ with $- \gauss{\Sarm/2} - 1 \leq i \le -1$ or $1 \le i \le \upgauss{\Sarm/2 }$.
\end{lemma}
\begin{proof}
	For every integer $i$ with $-\gauss{\Sarm/2} \le i \le -1$ or $1 \le i \le \upgauss{\Sarm/2} - 1$ (excluding $-\gauss{\Sarm/2} - 1$ and $\upgauss{\Sarm/2}$ as stated in the claim),
		the maximal inward and outward matches at the positions $m + i$ and $m + d - i$ yields $(\LL\arm, \RR\arm)$,
		and thus, $(m + i, d - 2i)$ cannot be in $\imgPali$ due to the injectivity of~$\mapPali$ (cf.\ \Cref{figPointsPalProp}).
		If $i$ is $- \gauss{\Sarm/2} - 1$ or $\lceil \Sarm/2 \rceil$,
		the point $(m + i, d - 2i) \in \Z^2$ is not in $\imgPali$ because the pair of positions $m+i$ and $m+d-i$ is where the inward or outward match from the positions $m$ and $m + i$ fails. %
\end{proof}

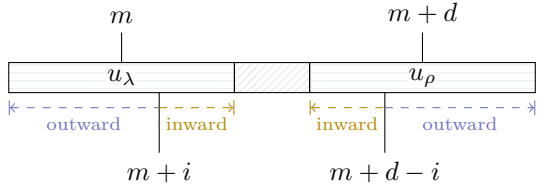
\begin{figure}[ht]
	\MyFloatBox{%
		\begin{tikzpicture}[yscale=0.4,xscale=0.5]
\tikzstyle{MyLabel} = [dashed,font=\scriptsize]
  \void{3}
  \markeup{$m$}
  \void{8}
  \markeup{$m+d$}
  \newrow
  \col[arm]{6}{$\LL\arm$}
  \col[gap]{2}{}
  \col[arm]{6}{$\RR\arm$}

  \newrow
	\draw [MyLabel,|<-,color=solarizedViolet] (\ourCol,\ourRow-0.5) -- (\ourCol+4,\ourRow-0.5) node [midway,anchor=north] {outward};
   \void{4}
  \marke{}
	\draw [MyLabel,->|,color=solarizedYellow] (\ourCol,\ourRow-0.5) -- (\ourCol+2,\ourRow-0.5) node [midway,anchor=north] {inward};
   \void{4}
	\draw [MyLabel,|<-,color=solarizedYellow] (\ourCol,\ourRow-0.5) -- (\ourCol+2,\ourRow-0.5) node [midway,anchor=north] {inward};
   \void{2}
  \marke{}
	\draw [MyLabel,->|,color=solarizedViolet] (\ourCol,\ourRow-0.5) -- (\ourCol+4,\ourRow-0.5) node [midway,anchor=north] {outward};
  \newrow
   \void{4}
  \marke{$m+i$}
   \void{4}
   \void{2}
  \marke{$m+d-i$}
	\end{tikzpicture}
	}{%
	\caption{Setting of \Cref{lemmaPointsPalProp}. A gapped palindrome~$(\LL\arm, \RR\arm)$ is mapped to the point~$(m,d)$. 
		It can be restored by longest common prefix and suffix queries at the positions $(m + i, d - 2i)$ for every integer $i$ with
		$- \gauss{\Sarm/2} - 1 \leq i \le -1$ or $1 \le i \le \upgauss{ \Sarm/2 }$, where $\Sarm := \abs{\LL\arm} = \abs{\RR\arm}$.
}
\label{figPointsPalProp}
	}%
\end{figure}

Due to \Cref{lemmaPointsPalProp}, each point $(m, d) \in \imgPali$ has at least one distinct point that is not in the image of~$\mapPali$.
For instance, we count the point $(m + 1, d - 2) \in \CoverN{n} \setminus \imgPali$ ($d \ge 3$ according to \Cref{corGappedCorrectionPlaces}\ref{itGappedCorrectionOrdinary}) for each $(m, d) \in \imgPali$, and each counted point is counted only once.
With this insight, we can prove the next \lcnamecref{lemmaPointsPali} in exactly the same way as \Cref{lemmaPointsRep}.

\begin{corollary}\label{lemmaPointsPali}
    Let $\gamma$ be a real number with $\gamma \in (0,1]$.
	A set of points $\SubCover \subseteq \imgPali$ such that no two distinct points in~$\SubCover$ $\gamma$-cover the same point obeys the inequality 
	$\abs{\SubCover} < n (\pi^2 / 6 - 1/2)/\gamma$.
\end{corollary}
\begin{proof}
With the same definition of~$\SetE$ as in the proof of \Cref{lemmaPointsRep}, it is left to show that the sum of the weights of 
all points in $\SetE$ is at most $n / (2\gamma)$.

Unlike the proof of \Cref{lemmaPointsRep}, we can take a shortcut with the following observation:
Only the highest points in~$\SetE$ can be $\gamma$-covered by a point from~$\SubCover \setminus \SetE$.\footnote{This holds also in the case of maximal $\alpha$-gapped repeats in the proof of \Cref{lemmaPointsRep}.
However, this trick does not lead to anything useful there.}
To see this, let $(x,y) \in \SetE$ be a point with $y < 1/\gamma - 1$.
Assume that $(\hat{x},\hat{y})$ $\gamma$-covers~$(x,y)$, then
$\hat{y} - \gamma \hat{y} \le y <  1/\gamma - 1$, or equivalently $\hat{y} < 1/\gamma$.
This means that $(x,y) = (\hat{x},\hat{y})$.

We conclude that every point~$(x,y) \in \SetE$ with $\weight[x,y] = 1$ belongs to~$\SubCover$, and therefore 
(a) $(x+1,y-2) \not\in \SubCover$ according to \Cref{lemmaPointsPalProp}, and (b) $\weight[x+1,y-2] = 0$ according to the above observation.
Hence, both points~$(x,y)$ and~$(x+1,y-2)$ have a total weight of~$1$ (remember that $\weight[n,y-2] = 0$ in any case, cf.\ proof of \Cref{lemmaPointsRep}).

Although a highest point~$(x,y)$ (with $1/\gamma - 1 \le y$) can be $\gamma$-covered by a point in $\SubCover \setminus \SetE$, 
one of its neighbors~$(x-1,y)$ or $(x+1,y)$ has to be $\gamma$-covered by the same point, 
such that the sum of the weights of both points is at most $1/2$.
The total weight of all points in $\SetE$ is therefore at most $(1/2) \abs{\SetE} \le n / (2\gamma)$.
\end{proof}

\Cref{lemmaPointsPali} finally leads us to the connection between the $\gamma$-cover property and the maximal $\alpha$-gapped palindromes:

\begin{lemma}\label{lemmaPalNPCover}
    Let~$\warm$ be a word, and $\alpha$ and $\beta$ two real numbers with $\alpha > 1$ and $6/7 \leq \beta < 1$.
    The points mapped by two different maximal gapped palindromes in $\gpNP$ cannot $\frac{1-\beta}{\alpha}$-cover the same point.
\end{lemma}
\begin{proof}
  Let $(\LL\arm,\RR\arm)$ and $(\s{\LL\arm},\s{\RR\arm})$ be two different maximal $\alpha$-gapped palindromes in $\gpNP$.
  Set $\Sarm := \abs{\LL\arm} = \abs{\RR\arm}$ and $\s\Sarm := \abs{\s{\LL\arm}} = \abs{\s{\RR\arm}}$.
  Let $(m, d)$ and $(\s{m},\s{d})$ be the points mapped from~$\agr$ and~$\sagr$, respectively.
  Assume, for the sake of contradiction, that both points $\frac{1-\beta}{\alpha}$-cover the same point $(x, y)$.

  Let $\zvar := \abs{m - \s{m}}$, and
  let $\LL\subarm := \LL\arm \cap \s{\LL\arm}$
  be the overlap of $\LL\arm$ and $\s{\LL\arm}$.
  Let $\Ssubarm := \abs{\LL\subarm}$, 
  and let $\RR\subarm$ (resp. $\s{\RR\subarm}$) be the reverse copy of $\LL\subarm$ based on $\agr$ (resp. $\sagr$), i.e.,
  $\LL\subarm = \pali{\RR\subarm} = \pali{\s{\RR\subarm}}$ with
  $\ibeg{\RR\subarm} = \ibeg{\RR\arm}+\iend{\RR\arm}-\iend{\LL\subarm}$
  and
  $\ibeg{\s{\RR\subarm}} = \ibeg{\s{\RR\arm}}+\iend{\s{\RR\arm}}-\iend{\LL\subarm}$.

  \begin{subclaim}  
The overlap~$\LL\subarm$ is not empty, and $\ibeg{\RR\subarm} \not= \ibeg{\s{\RR\subarm}}$.
\end{subclaim}

  \begin{subproof}
      First we show that $\LL\subarm$ is not empty.
      If $m = \s{m}$, it is clear that $\LL\subarm$ contains $\warm[m]$.
      Without loss of generality, assume that $m < \s{m}$ for this sub-proof (otherwise exchange $\agr$ with $\sagr$).
	  By combining (a) the $(1-\beta)/\alpha$-cover property with (b) the fact that $\sagr$ is $\alpha$-gapped and (c) the constraint $6/7 \leq \beta < 1$, we obtain
      \(
        \s{m} - \s\Sarm / 2
		\le_{\text{(c)}} \s{m} - (1-\beta)\s\Sarm
		\le_{\text{(b)}} \s{m} - \s{d}(1-\beta)/\alpha 
		\le_{\text{(a)}} x \le_{\text{(a)}} m < \s{m}.
      \)
	  This long inequality says that the text position~$m$ is contained in $\s{\LL\arm}$, which implies that $\LL\subarm$ is not empty.
	  If $\RR\subarm$ and $\s{\RR\subarm}$ start at the same position, then expanding the arms~$\LL\subarm$ and $\RR\subarm (\equiv \s{\RR\subarm})$ to the left and right
	  yields the arms $\LL\arm \equiv \s{\LL\arm}$ and $\RR\arm \equiv \s{\RR\arm}$, which implies that $\agr$ and $\sagr$ are the same gapped repeat, a contradiction.
  \end{subproof}

  Without loss of generality let $d \le \s{d}$.
  With the $(1-\beta)/\alpha$-cover property we obtain
  \begin{equation}\label{equOrdGapRuleInvPal}
    \s{d} - \frac{\s{d}(1-\beta)}{\alpha} \le y \le d \le \s{d}.
  \end{equation}
  The difference $\delta :=  \s{d} - d \ge 0$ can be estimated by
  \begin{equation}\label{equOrdGapRulePal}
    \delta \le \s{d}(1-\beta)/\alpha \le \s\Sarm(1-\beta).
  \end{equation}
  \Cref{equOrdGapRuleInvPal} can also be used to lower bound $\Sarm$ in terms of $\s{d}$ due to the fact that~$\agr$ is $\alpha$-gapped:
  \begin{equation}\label{equOrdGapRuleTakeFirstArmPal}
  	\Sarm \ge d/\alpha \ge \frac{\s{d}}{\alpha} (1 - \frac{1-\beta}{\alpha}) \ge \s{d}\beta/\alpha.
  \end{equation}

  \block{Outline}
  In the following we conduct a thorough case analysis.
  In each case we show the contradiction that $\agr$ or $\sagr$ is $\beta$-periodic.
  We prove each case in a similar way:
  We first show that the intersection of~$\RR\subarm$ and $\s{\RR\subarm}$ is large enough such that it induces a repetition on $\RR\subarm \cup \s{\RR\subarm}$.
  Subsequently, we find a run covering $\RR\subarm \cup \s{\RR\subarm}$, and another run covering $\LL\subarm$.
  However, since $\LL\subarm$ is the suffix of $\LL\arm$ (resp.\ $\s{\LL\arm}$), we can conclude that $\agr$ (resp.\ $\sagr$) is $\beta$-periodic.

  Before starting with the case analysis, we introduce a general property of the starting positions~$\ibeg{\RR\subarm}$ and~$\ibeg{\s{\RR\subarm}}$
  needed for the analysis.
  Adding up the equalities of \Cref{corGappedCorrectionPlaces}\ref{itGappedCorrectionPlacesELarm} and~\ref{itGappedCorrectionBRArm} gives
$\ibeg{\RR\arm} + \iend{\LL\arm} = 2m + d$. 
With that we obtain $\ibeg{\RR\subarm} = \ibeg{\RR\arm} + \iend{\LL\arm} - \iend{\LL\subarm} = 2m + d - \iend{\LL\subarm}$.
Hence, the distance between the starting positions of~$\RR\subarm$ and~$\s{\RR\subarm}$ is given by 
\begin{equation}\label{equOrdStartSubarmPal}
  \abs{\ibeg{\RR\subarm} - \ibeg{\s{\RR\subarm}}} =
\begin{cases}
	2\zvar + \delta & \text{if~} m \le \s{m},  \\
	2\zvar - \delta & \text{if~} m > \s{m} \text{~and~} \ibeg{\RR\subarm} > \ibeg{\s{\RR\subarm}}, \text{~or} \\
	\delta - 2\zvar & \text{if~} m > \s{m} \text{~and~} \ibeg{\RR\subarm} <  \ibeg{\s{\RR\subarm}}.
\end{cases}
\end{equation}

  \begin{figure}[ht]
    \MyFloatBox{%
		\begin{tikzpicture}[yscale=0.4,xscale=0.92]
  \col[arm]{4.5}{$\LL\arm$}
  \col[gap]{1.5}{}
  \col[arm]{4.5}{$\RR\arm$}
  \newrow
  \void{1}
  \col[arm]{3}{$\s{\LL\arm}$}
  \col[gap]{4}{}
  \col[arm]{3}{$\s{\RR\arm}$}
  \newrow
  \void{1}
  \col[subarm]{3}{$\LL\subarm$}
  \void{2.5}
  \col[subarm]{3}{$\RR\subarm$}
  \newrow
  \void{1}
  \col[rep]{3}{$\rep$}
  \void{4}
  \col[subarm]{3}{$\s{\RR\subarm}$}
	\end{tikzpicture}
		}{%
	\caption{\ref{GapPaliSubCase1A} in the proof of \Cref{lemmaPalNPCover} with $m \le \s{m}$ and $\ibeg{\LL\arm} \le \ibeg{\s{\LL\arm}} \le \iend{\s{\LL\arm}} \le \iend{\LL\arm}$.}
    \label{figPaliSubOneA}
		}%
  \end{figure}
  \GappedCase{Case~1}{$m \le \s{m}$}
  Since $\s{m} - \s{d}(1-\beta)/\alpha \le x \le m \le \s{m}$ (due to the $(1-\beta)/\alpha$-cover property),
  \begin{equation}\label{equOrdGapRulePosFirstLeftPal}
	  \zvar = \s{m} - m \le \s{d}(1-\beta)/\alpha \le \s\Sarm(1-\beta),
  \end{equation}
  because $\sagr$ is $\alpha$-gapped.
  Due to \cref{equOrdStartSubarmPal}, the starting positions of both right copies $\s{\RR\subarm}$ and $\RR\subarm$ differ by $\ibeg{\s{\RR\subarm}} - \ibeg{\RR\subarm} = 2\zvar + \delta > 0$.
  By~\cref{equOrdGapRulePal,equOrdGapRulePosFirstLeftPal}, we get
  
  \begin{equation}\label{equPropZPal}
	  2\zvar + \delta \le 3 \s{d}(1-\beta)/\alpha \le 3 \s\Sarm(1-\beta).
  \end{equation}

  Depending on the relations $\ibeg{\LL\arm} \lesseqgtr \ibeg{\s{\LL\arm}}$ and $\iend{\LL\arm} \lesseqgtr \iend{\s{\LL\arm}}$, 
  we split the case in four sub-cases.
  However, one of the four sub-cases with $\ibeg{\s{\LL\arm}} < \ibeg{\LL\arm}$ and $\iend{\s{\RR\arm}} < \iend{\RR\arm}$ already leads to a contradiction (without proving that one left arm has a periodic suffix):
  Assume that both inequalities $\ibeg{\s{\LL\arm}} < \ibeg{\LL\arm}$ and $\iend{\s{\RR\arm}} < \iend{\RR\arm}$ hold for the sake of contradiction.
  Under these assumptions, with \Cref{corGappedCorrectionPlaces}\ref{itGappedCorrectionPlacesBLArm} it must hold that $\ibeg{\s{\LL\arm}} + 1/2 = \s{m} - \corr{\s\Sarm} - \s\Sarm / 2 + 1 \leq m - \corr{\Sarm} - \Sarm / 2 = \ibeg{\LL\arm}-1/2$
  and $\iend{\s{\RR\arm}} + 1/2 = \s{m} - \corr{\s\Sarm} + \s\Sarm / 2 + 1 \leq m - \corr{\Sarm} + \Sarm / 2 = \iend{\RR\arm} - 1/2$.
  Adding the left sides and the right sides of both inequalities gives $\s{m} - m \leq \corr{\s\Sarm} - \corr{\Sarm} - 1 < 0$,
  which contradicts that $\s{m} - m \geq 0$.

  Thus, it is enough to consider the following three sub-cases 1a, 1b, and 1c.

  \GappedSubCase{\CustomLabel{GapPaliSubCase1A}{Sub-Case~1a}}{$\LL\subarm \equiv \s{\LL\arm}$}{figPaliSubOneA}
  Since $\s\Sarm / (2\zvar + \delta) \geq \s\Sarm / (3 \s\Sarm(1 - \beta)) \geq 7/3 > 2$ holds (due to \cref{equPropZPal}) for $6/7 \leq \beta < 1$,
  we conclude that $\RR\subarm = \s{\RR\arm}$ is periodic, which means that $\sagr \in \gpP$, a contradiction.

  \begin{figure}[ht]
    \MyFloatBox{%
		\begin{tikzpicture}[yscale=0.4,xscale=0.92]
  \void{0.5}
  \col[arm]{3}{$\LL\arm$}
  \col[gap]{2}{}
  \col[arm]{3}{$\RR\arm$}
  \newrow
  \col[arm]{4.5}{$\s{\LL\arm}$}
  \col[gap]{1.5}{}
  \col[arm]{4.5}{$\s{\RR\arm}$}
  \newrow
  \void{0.5}
  \col[subarm]{3}{$\LL\subarm$}
  \void{2}
  \col[subarm]{3}{$\RR\subarm$}
  \newrow
  \void{0.5}
  \col[rep]{3}{$\rep$}
  \void{2}
  \col[diff]{1.5}{$2\zvar+\delta$}
  \col[subarm]{3}{$\s{\RR\subarm}$}
  \newrow
  \void{5.5}
	\end{tikzpicture}
		}{%
	\caption{\ref{GapPaliSubCase1B} in the proof of \Cref{lemmaPalNPCover} with $m \le \s{m}$ and $\ibeg{\s{\LL\arm}} \le \ibeg{\LL\arm} \le \iend{\LL\arm} \le \iend{\s{\LL\arm}}$.}
    \label{figPaliSubOneB}
		}%
  \end{figure}
  \GappedSubCase{\CustomLabel{GapPaliSubCase1B}{Sub-Case~1b}}{$\LL\subarm \equiv \LL\arm$}{figPaliSubOneB} 
  Recall that $\Sarm = \Ssubarm \ge \s{d} \beta / \alpha$ by~\cref{equOrdGapRuleTakeFirstArmPal}.
  It follows from \cref{equPropZPal} and $6/7 \leq \beta < 1$ 
  that $\Ssubarm / (2\zvar + \delta) \ge \s{d} \alpha \beta / (3 \s{d} \alpha (1-\beta)) = \beta / (3(1-\beta)) \geq 2$.
  Hence $\RR\subarm \equiv \RR\arm$ is periodic, which means that $\agr \in \gpP$, a contradiction.

  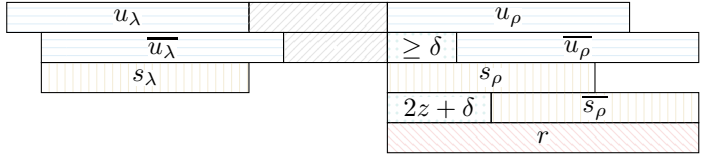
\begin{figure}[ht]
	\MyFloatBox{%
		\begin{tikzpicture}[yscale=0.4,xscale=0.92]
  \col[arm]{3.5}{$\LL\arm$}
  \col[gap]{2}{}
  \col[arm]{3.5}{$\RR\arm$}
  \newrow
  \void{0.5}
  \col[arm]{3.5}{$\s{\LL\arm}$}
  \col[gap]{1.5}{}
  \col[diff]{1}{$\ge \delta$}
  \col[arm]{3.5}{$\s{\RR\arm}$}
  \newrow
  \void{0.5}
  \col[subarm]{3}{$\LL\subarm$}
  \void{2}
  \col[subarm]{3}{$\RR\subarm$}
  \newrow
  \void{5.5}
  \col[diff]{1.5}{$2\zvar+\delta$}
  \col[subarm]{3}{$\s{\RR\subarm}$}
  \newrow
  \void{5.5}
  \col[rep]{4.5}{$\rep$}
	\end{tikzpicture}
		}{%
	\caption{\ref{GapPaliSubCase1C} in the proof of \Cref{lemmaPalNPCover} with $m \le \s{m}$ and $\ibeg{\LL\arm} < \ibeg{\s{\LL\arm}} \le \iend{\LL\arm} < \iend{\s{\LL\arm}}$.
	The second inequality holds because the overlap~$\LL\subarm$ cannot be empty due to the sub-claim.}
    \label{figPaliSubOneC}
		}%
  \end{figure}
  \GappedSubCase{\CustomLabel{GapPaliSubCase1C}{Sub-Case~1c}}{$\ibeg{\LL\arm} < \ibeg{\s{\LL\arm}}$ and $\iend{\LL\arm} < \iend{\s{\LL\arm}}$}{figPaliSubOneC}
  Since $\LL\subarm$ is a suffix of $\LL\arm$ and a prefix of $\s{\LL\arm}$, 
  the reverse copies~$\RR\subarm$ and $\s{\RR\subarm}$ are a prefix of $\RR\arm$ and a suffix of $\s{\RR\arm}$, respectively.
  We have $1 \le \ibeg{\s{\LL\arm}} - \ibeg{\LL\arm} = \s{m} - \s\Sarm/2 - \corr{\s\Sarm} - (m - \Sarm/2 - \corr{\Sarm})$. A simple reshaping leads to
  $\s{m} - \s\Sarm/2 - (m - \Sarm/2) \geq 1 + \corr{\s\Sarm} - \corr{\Sarm}$.
  This inequality yields 
  $\ibeg{\s{\RR\arm}} - \ibeg{\RR\arm} = \s{m} + \s{d} - \s\Sarm/2 + \corr{\s\Sarm} - (m + d - \Sarm/2 + \corr{\Sarm}) \geq \delta + 1 + 2(\corr{\s\Sarm} - \corr{\Sarm}) \geq \delta \geq 0$.
  This means that
  $\ibeg{\RR\subarm} = \ibeg{\RR\arm} \le \ibeg{\s{\RR\arm}} \le \ibeg{\s{\RR\subarm}} \le \iend{\s{\RR\subarm}} = \iend{\s{\RR\arm}}$.
  With $\ibeg{\RR\subarm} \le \ibeg{\s{\RR\arm}} = \iend{\s{\RR\subarm}} - \s{\Sarm} + 1 = 
  \ibeg{\s{\RR\subarm}} + \Ssubarm - \s{\Sarm}$, it follows that
  $\Ssubarm \ge \s{\Sarm} - (2\zvar+\delta) > 2\zvar+\delta$ because
  $\s{\Sarm} / (2\zvar + \delta) \geq \s\Sarm / (3 \s\Sarm(1 - \beta)) \geq 7/3 > 2$ holds for $6/7 \leq \beta < 1$.
  Since $\ibeg{\s{\RR\subarm}} - \ibeg{\RR\subarm} = \iend{\s{\RR\subarm}} - \iend{\RR\subarm} = 2\zvar+\delta$,
  $\RR\subarm \cap \s{\RR\subarm} \not= \emptyset$.
This means that $\s{\RR\arm} \subset \RR\subarm \cup \s{\RR\subarm}$, and that $\s{\RR\arm}$ is periodic with a period of at most $2\zvar+\delta$, a contradiction.

  \GappedCase{Case~2}{$m > \s{m}$}
  Since $m - d (1-\beta)/\alpha \le x \le \s{m} < m$,
  \begin{equation}\label{equOrdGapRulePosFirstRightPal}
    \zvar = m - \s{m} \le d(1-\beta)/\alpha \le \s{d}(1-\beta)/\alpha \le \s\Sarm(1-\beta).
  \end{equation}
  Due to \cref{equOrdStartSubarmPal}, the starting positions of both right copies differ by $\abs{\ibeg{\RR\subarm} - \ibeg{\s{\RR\subarm}}} = \abs{2\zvar - \delta}$.
  \Cref{equOrdStartSubarmPal}
  with \cref{equOrdGapRulePal,equOrdGapRulePosFirstRightPal} yields
  \begin{equation}\label{equOrdGapRuleDistSubArmPal}
\abs{2\zvar - \delta} \leq 
  \begin{cases}
	  2\zvar \le 2 \s{d}(1-\beta)/\alpha \le 2 \s\Sarm (1 - \beta) & \text{if~} \ibeg{\RR\subarm} > \ibeg{\s{\RR\subarm}}, \text{~or} \\
	  \delta \le \s{d}(1-\beta)/\alpha \le \s\Sarm (1 - \beta)  & \text{if~} \ibeg{\RR\subarm} < \ibeg{\s{\RR\subarm}}.
  \end{cases}
  \end{equation} 

  We split again the case into sub-cases depending on the relation of the starting and of the ending positions of the left arms.
  The sub-case with $\ibeg{\LL\arm} < \ibeg{\s{\LL\arm}}$ and $\iend{\RR\arm} < \iend{\s{\RR\arm}}$ already leads to a contradiction, which can be seen 
  by an argument that is similar to the one used in Case~1 due to symmetry.

  \begin{figure}[ht]
    \MyFloatBox{%
		\begin{tikzpicture}[yscale=0.4,xscale=0.92]
  \col[arm]{5}{$\LL\arm$}
  \col[gap]{2.5}{}
  \col[arm]{5}{$\RR\arm$}
  \newrow
  \void{0.5}
  \col[arm]{3.5}{$\s{\LL\arm}$}
  \col[gap]{3}{}
  \col[arm]{3.5}{$\s{\RR\arm}$}
  \newrow
  \void{0.5}
  \col[subarm]{3.5}{$\LL\subarm$}
   \void{3}
   \col[diff]{1.5}{$\abs{2\zvar-\delta}$}
  \col[subarm]{3.5}{$\RR\subarm$}
  \newrow
  \void{0.5}
  \col[rep]{3.5}{$\rep$}
  \void{3}
  \col[subarm]{3.5}{$\s{\RR\subarm}$}
	\end{tikzpicture}
		}{%
	\caption{\ref{GapPaliSubCase2A} in the proof of \Cref{lemmaPalNPCover} with $m > \s{m}$ and $\ibeg{\LL\arm} \le \ibeg{\s{\LL\arm}} \le \iend{\s{\LL\arm}} \le \iend{\LL\arm}$.}
    \label{figPaliSubTwoA}
		}%
  \end{figure}
  \GappedSubCase{\CustomLabel{GapPaliSubCase2A}{Sub-Case~2a}}{$\LL\subarm \equiv \s{\LL\arm}$}{figPaliSubTwoA}
  Since $\Ssubarm / \abs{2\zvar - \delta} \ge \s\Sarm / (2 \s\Sarm (1-\beta)) = 1 / (2(1-\beta)) \ge 7/2 > 2$ holds (due to \cref{equOrdGapRuleDistSubArmPal}) for $6/7 \leq \beta < 1$,
  the distance between $\ibeg{\RR\subarm}$ and $\ibeg{\s{\RR\subarm}}$ is small enough such that 
  $\RR\subarm = \s{\RR\arm}$ is periodic, which means that $\sagr \in \gpP$, a contradiction.

  \begin{figure}[ht]
    \MyFloatBox{%
		\begin{tikzpicture}[yscale=0.4,xscale=0.92]
  \void{1}
  \col[arm]{3}{$\LL\arm$}
  \col[gap]{3}{}
  \col[arm]{3}{$\RR\arm$}
  \newrow
  \col[arm]{4.5}{$\s{\LL\arm}$}
  \col[gap]{1.5}{}
  \col[arm]{4.5}{$\s{\RR\arm}$}
  \newrow
  \void{1}
  \col[subarm]{3}{$\LL\subarm$}
  \void{3}
  \col[subarm]{3}{$\RR\subarm$}
  \newrow
  \void{1}
  \col[rep]{3}{$\rep$}
  \void{2.5}
  \col[subarm]{3}{$\s{\RR\subarm}$}
  \newrow
	\end{tikzpicture}
		}{%
	\caption{\ref{GapPaliSubCase2B} in the proof of \Cref{lemmaPalNPCover} with $m > \s{m}$ and $\ibeg{\s{\LL\arm}} \le \ibeg{\LL\arm} \le \iend{\LL\arm} \le \iend{\s{\LL\arm}}$.}
    \label{figPaliSubTwoB}
		}%
  \end{figure}
  \GappedSubCase{\CustomLabel{GapPaliSubCase2B}{Sub-Case~2b}}{$\LL\subarm \equiv \LL\arm$}{figPaliSubTwoB}
  Recall that $\Sarm = \Ssubarm \ge \s{d} \beta / \alpha$ by~\cref{equOrdGapRuleTakeFirstArmPal}.
  It follows from $6/7 \leq \beta < 1$ and \cref{equOrdGapRuleDistSubArmPal}
  that $\Ssubarm / \abs{2\zvar - \delta} \ge \s{d} \alpha \beta / (2 \s{d} \alpha (1-\beta)) = \beta / (2(1-\beta)) \geq 3 > 2$.
  Hence $\RR\subarm \equiv \RR\arm$ is periodic, which means that $\agr \in \gpP$, a contradiction.

  \GappedCase{\CustomLabel{GapPaliSubCase2C}{Sub-Case~2c}}{$\ibeg{\LL\arm} > \ibeg{\s{\LL\arm}}$ and $\iend{\LL\arm} > \iend{\s{\LL\arm}}$} 
  Since $\LL\subarm$ is a prefix of $\LL\arm$ and a suffix of $\s{\LL\arm}$, 
  the reverse copies~$\RR\subarm$ and $\s{\RR\subarm}$ are a suffix of $\RR\arm$ and a prefix of $\s{\RR\arm}$, respectively.

  \begin{figure}[t]
    \MyFloatBox{%
        	\begin{tikzpicture}[yscale=0.4,xscale=1.2]
  \void{0.5}
  \col[arm]{2.5}{$\LL\arm$}
  \col[gap]{0.75}{}
  \col[diff]{0.5}{$\ge 0$}
  \col[arm]{2.5}{$\RR\arm$}
  \newrow
  \col[arm]{2.25}{$\s{\LL\arm}$}
  \col[gap]{1.5}{}
  \col[arm]{2.25}{$\s{\RR\arm}$}
  \newrow
  \void{0.5}
  \col[subarm]{1.75}{$\LL\subarm$}
  \void{1.5}
  \col[diff]{1.25}{$\abs{2\zvar-\delta}$}
  \col[subarm]{1.75}{$\RR\subarm$}
  \newrow
  \void{3.75}
  \col[subarm]{1.75}{$\s{\RR\subarm}$}
  \newrow
  \void{3.75}
  \col[rep]{3}{$\rep$}
        \end{tikzpicture}
		}{%
	\caption{\ref{GapPaliSubSubCase2CI} in the proof of \Cref{lemmaPalNPCover} with $m > \s{m}$, $\ibeg{\s{\LL\arm}} < \ibeg{\LL\arm} \le \iend{\s{\LL\arm}} < \iend{\LL\arm}$ and $\ibeg{\s{\RR\arm}} \le \ibeg{\RR\arm}$.}
    \label{figPaliSubTwoCi}
		}%
  \end{figure}
  \GappedSubCase{\CustomLabel{GapPaliSubSubCase2CI}{Subsub-Case~2c-i}}{$\ibeg{\RR\arm} \geq \ibeg{\s{\RR\arm}}$}{figPaliSubTwoCi}
  Recall that $\Sarm \ge \s{d} \beta / \alpha$ by~\cref{equOrdGapRuleTakeFirstArmPal}.
  It follows from $6/7 \leq \beta < 1$ and \cref{equOrdGapRuleDistSubArmPal}
  that $\Sarm / \abs{2\zvar - \delta} \ge \s{d} \alpha \beta / (2 \s{d} \alpha (1-\beta)) = \beta / (2(1-\beta)) \geq 3 > 2$.
  With $\ibeg{\RR\arm} \geq \ibeg{\s{\RR\arm}}$, this case is symmetric to~\ref{GapPaliSubCase1C},
  leading to the result that $\RR\arm \subset \RR\subarm \cup \s{\RR\subarm}$, 
  and that $\RR\arm$ is periodic with a period of at most $\abs{2\zvar-\delta}$, a contradiction.

  \begin{figure}[ht]
    \MyFloatBox{%
		\begin{tikzpicture}[yscale=0.4,xscale=1.3]
  \void{0.5}
  \col[arm]{3}{$\LL\arm$}
  \col[gap]{0.5}{}
  \col[arm]{3}{$\RR\arm$}
  \newrow
  \col[arm]{3}{$\s{\LL\arm}$}
  \col[gap]{2}{}
  \col[arm]{3}{$\s{\RR\arm}$}
  \newrow
  \col[diff]{0.5}{$\le \delta$}
  \col[subarm]{2.5}{$\LL\subarm$}
  \void{1.5}
  \col[subarm]{2.5}{$\RR\subarm$}
  \newrow
  \void{0.5}
  \col[rep]{2.5}{$\rep$}
  \void{2}
  \col[subarm]{2.5}{$\s{\RR\subarm}$}
  \col[diff]{0.5}{$\le \delta$}
  \newrow
  \newrow
  \void{4.5}
\dist[$\abs{2\zvar-\delta}$]{0.5}
	\end{tikzpicture}
		}{%
	\caption{\ref{GapPaliSubSubCase2CII} in the proof of \Cref{lemmaPalNPCover} with $m > \s{m}$, 
	$\ibeg{\s{\LL\arm}} < \ibeg{\LL\arm} \le \iend{\s{\LL\arm}} < \iend{\LL\arm}$ and $\ibeg{\RR\arm} < \ibeg{\s{\RR\arm}}$.}
    \label{figPaliSubTwoCii}
		}%
  \end{figure}
  \GappedSubCase{\CustomLabel{GapPaliSubSubCase2CII}{Subsub-Case~2c-ii}}{$\ibeg{\RR\arm} < \ibeg{\s{\RR\arm}}$}{figPaliSubTwoCii}
  It follows from $\ibeg{\RR\arm} < \ibeg{\s{\RR\arm}}$ that
  $\ibeg{\RR\arm} - \ibeg{\s{\RR\arm}} = m + d - \Sarm/2 + \corr{\Sarm} - (\s{m} + \s{d} - \s\Sarm/2 + \corr{\s\Sarm})
  = \ibeg{\LL\arm} - \ibeg{\s{\LL\arm}} - \delta + 2(\corr{\Sarm} - \corr{\s\Sarm}) \leq -1$, which leads to
  $\ibeg{\LL\arm} - \ibeg{\s{\LL\arm}} \leq \delta + 2(\corr{\s\Sarm} - \corr{\Sarm}) - 1 \leq \delta$.
  Combining this inequality with \cref{equOrdGapRulePal} gives
  $\Ssubarm = \s\Sarm - (\ibeg{\LL\arm} - \ibeg{\s{\LL\arm}}) \geq \s\Sarm - \delta \geq \beta \s\Sarm$.
  With \cref{equOrdGapRuleDistSubArmPal} this yields $\Ssubarm / \abs{2\zvar - \delta} \ge \beta \s\Sarm / (2\s\Sarm (1-\beta)) = \beta / (2(1-\beta)) \ge 3 > 2$ under the presumption that $6/7 \leq \beta < 1$.
  This means that $\s{\LL\arm}$ has a periodic suffix of length $\beta \s\Sarm$, and that $\sagr \in \gpP$, a contradiction.
\end{proof}

Combining the results of \cref{lemmaPointsPali} and \cref{lemmaPalNPCover} immediately gives the following \lcnamecref{lemmaPalNP}:
\begin{corollary}\label{lemmaPalNP}
	Given two real numbers $\alpha$ and $\beta$ with $\alpha > 1$ and $7/9 \leq \beta < 1$, and a 
	word~$\warm$ of length~$n$, 
	the number of all maximal $\alpha$-gapped $\beta$-aperiodic palindromes is bounded by the inequality
	$\abs{\gpNP} < \alpha n (\pi^2 / 6 - 1/2)/(1-\beta)$.
\end{corollary}

\begin{theorem}\label{thmMaxPal}
	Given a real number $\alpha$ with $\alpha > 1$, and a word~$\warm$ of length~$n$, 
	the number of all maximal $\alpha$-gapped palindromes~$\abs{\gpGR}$ less than
	$\PalindromeResult$. 
\end{theorem}
\begin{proof}
	Combining the results of \cref{lemmaP} and \cref{lemmaPalNP} yields
  \[
  \abs{\gpGR} 
  = \underbrace{2n - 1}_{\text{max.\ palindromes}} + \underbrace{\abs{\gpP}}_{\text{$\beta$-periodic}} + \underbrace{\abs{\gpNP}}_{\text{$\beta$-aperiodic}} 
  < 2n - 1 + \underbrace{2 (\alpha-1) \frac{\sumExp{\warm}}{\beta}}_{\text{\Cref{lemmaP}}} + \underbrace{(\frac{\pi^2}{6} - \frac{1}{2}) \frac{\alpha n}{1 - \beta}}_{\text{\Cref{lemmaPalNP}}}
  \]
  for every $6/7 \leq \beta < 1$.
  Applying~\cref{lemmaExponent}, the term on the right side is upper bounded by 
  $2n - 1 + 2 (\alpha-1) (3n / \beta) + (\pi^2 / 6 - 1/2) \alpha n / (1 - \beta)$.
  This number is minimal when $\beta = 6/7$, yielding the bound 
  $2n - 1 + 7 n (\alpha - 1) + 7 (\pi^2 / 6 - 1/2) \alpha n = \PalindromeResult$.
\end{proof}

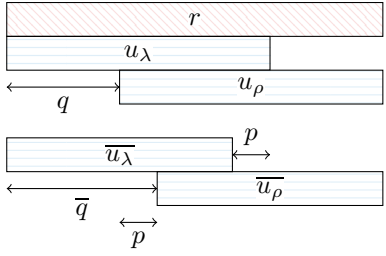
\begin{figure}[ht]
	\MyFloatBox{%
		\begin{tikzpicture}[yscale=0.45,xscale=0.5]
			\col[rep]{10}{$\rep$}
			\newrow
  \col[arm]{7}{$\LL\arm$}
			\newrow
			\void{3}
  \col[arm]{7}{$\RR\arm$}
  \newrow
			\dist[$\qvar$]{3}
  \newrow
  \col[arm]{6}{$\s{\LL\arm}$}
			\distup[$\period$]{1}
			\newrow
			\void{4}
  \col[arm]{6}{$\s{\RR\arm}$}
  \newrow
  \dist[$\s{\qvar}$]{4}
  \newrow
  \void{3}
			\dist[$\period$]{1}

	\end{tikzpicture}
		}{%
\caption{%
	Two gapped repeats~$(\LL\arm,\RR\arm)$ and~$(\s{\LL\arm},\s{\RR\arm})$ with overlapping arms.
	Both gapped repeats are within a run~$\rep$.
	They are maximal if their arms border the run~$\rep$.
	Each such maximal gapped repeat with overlapping arms has a period ($\qvar$ or $\s{\qvar}$ in the figure) that is a multiple of~$\rep$'s period~$\period$.
  }
\label{figMaxGappedOverlappingRepeats}
}%
\end{figure}

\section{A Linear Time Algorithm on Integer Alphabets}\label{secAlgoGappedRepeatOverlap}
In this algorithmic section, we are given a word~$\warm$ of length~$n$ on an \emph{integer} alphabet~$\Sigma$ as input such that $\abs{\Sigma} = n^{\Oh{1}}$.
In the following, we provide an \Oh{n} time algorithm that finds all maximal $\alpha$-gapped repeats/palindromes~$(\LL\arm,\RR\arm)$ 
with $\iend{\LL\arm} \ge \ibeg{\RR\arm}$. We call these $\alpha$-gapped repeats/palindromes \intWort{with overlap}.
We compute the other $\alpha$-gapped repeats/palindromes with a slight modification of the algorithm in~\cite{gawrychowski18tighter}, which
finds all maximal $\alpha$-gapped repeats/palindromes~$(\LL\arm,\RR\arm)$ with $\iend{\LL\arm} < \ibeg{\RR\arm}$, i.e., 
with a non-negative gap between~$\LL\arm$ and~$\RR\arm$.

When studying $\alpha$-gapped repeats/palindromes with overlap,
we can neglect the parameter~$\alpha$, because a gapped repeat/palindrome~$(\LL\arm,\RR\arm)$ whose arms overlap obeys 
the inequality~$\ibeg{\RR\arm} - \ibeg{\LL\arm} < \abs{\LL\arm} \le \alpha \abs{\LL\arm}$ for every~$\alpha \ge 1$.
For a gapped palindrome~$(\LL\arm,\RR\arm)$ with~$\iend{\LL\arm} \ge \ibeg{\RR\arm}$, we already know that either~$\agr$ is not maximal, 
or $\LL\arm \equiv \RR\arm$. 
Hence, a maximal gapped palindrome with an overlap is equal to a maximal ordinary palindrome.
It is well known that maximal ordinary palindromes can be found in \Oh{n} time~\cite{ManacherPalindromes}.

In what follows, we focus on the maximal gapped repeats with overlap.
Given a maximal gapped repeat~$(\LL\arm,\RR\arm)$ with period~$\qvar := \ibeg{\RR\arm} - \ibeg{\LL\arm} < \abs{\LL\arm}$,
it induces a square with $\warm[\ibeg{\LL\arm}..\ibeg{\LL\arm}+\qvar-1] = \warm[\ibeg{\RR\arm}..\ibeg{\RR\arm}+\qvar-1]$.
The square induces a run~$\rep$ whose minimal period~$\period$ divides~$\qvar$~(also observed in \cite[Conclusions]{crochemore16optimal}).
Both arms~$\LL\arm$ and~$\RR\arm$ are contained in~$\rep$.
Because $\agr$ is maximal, $\ibeg{\LL\arm} = \ibeg{\rep}$ and $\iend{\RR\arm} = \iend{\rep}$ hold; 
otherwise we could extend the arms to the left or to the right, respectively.
This means that the left arm~$\LL\arm$ covers at least the segment $\warm[\ibeg{\rep}..\ibeg{\rep}+\exp(\rep)\period/2]$ (otherwise the arms would not overlap).
Since~$\qvar$ is a multiple of~$\period$, the number of different lengths of~$\LL\arm$ is bounded by $\exp(\rep)/2$.
\Cref{figMaxGappedOverlappingRepeats} illustrates two maximal gapped repeats with overlapping arms within the same run.

Our idea is that we probe at the borders of each run~$\rep$ for all possible values of~$\qvar$ to find a gapped repeat
whose arms overlap and are contained in~$\rep$.
Having the \LCEds{} of~\cite{gawrychowski18tighter}, we spend~$\Oh{\exp(\rep)}$ time on each run~$\rep$,
summing up to $\Oh{n}$ due to \Cref{lemmaExponent}.
The positions of the runs can be computed in linear time~\cite{kolpakov99maximalRepetitions,runstheorem}. %
Since a gapped repeat~$(\LL\arm,\RR\arm)$ with overlapping arms is uniquely defined by its period and the borders of the run containing~$\LL\arm$ and $\RR\arm$, 
we can report each such gapped repeat exactly once.

Finally, it is left to modify the algorithm of \citet{gawrychowski18tighter} to find \emph{only} all \emph{maximal} $\alpha$-gapped repeats.
This modification is necessary, because a maximal gapped repeat in the scenario prohibiting overlaps is in general not a maximal gapped repeat
in the scenario supporting overlaps. 
Remembering $\warm = \Char{aaa}$ of \Cref{exGappedAAA}, 
it contains two maximal gapped repeats (with arm-length~one) when prohibiting overlaps, whereas
$\warm$ contains only one maximal gapped repeat (with arm-length~two) when supporting overlaps.
The modification is easy:
On reporting a gapped repeat, we additionally check whether its arms can be extended to the left or to the right with an LCE query.
In the case that we can extend both arms, we discard the gapped repeat instead of reporting it 
(the repeat would not be maximal without being extended, and the maximal gapped repeats with overlap are found with the above algorithm).
The algorithm finding all maximal $\alpha$-gapped palindromes can be changed analogously by discarding each discovered 
gapped palindrome whose inward extension results in an overlap of both arms.

\begin{theorem}
	Given a word~$\warm$ of length~$n$ on an integer alphabet, we can compute all maximal $\alpha$-gapped repeats~$\grGR$ and all maximal $\alpha$-gapped palindromes~$\gpGR$ 
	in \Oh{\alpha n} time.
\end{theorem}

\section{Conclusion}
We provided a thorough analysis on the maximum number of all maximal $\alpha$-gapped repeats and palindromes,
for which we achieved the bounds of \RepeatResult{} and \PalindromeResult{}, respectively, for a word of length~$n$.
Our proofs work for both supporting overlaps and prohibiting overlaps, and thus generalize the analysis of former studies.
Our study does not lead to a blind end, as can be seen by the following open problems:

\block{Generalizing Gaps}
A generalization of $\alpha$-gapped repeats are $(f,g)$-gapped repeats, i.e., gapped repeats~$(\LL\arm,\RR\arm)$ with the 
  additional property that $g(\abs{\LL\arm}) \le \ibeg{\RR\arm}-\iend{\LL\arm}-1 \le f(\abs{\LL\arm})$ for two functions $f,g : \N \rightarrow \RealNumber$.
  The $(f,g)$-gapped repeats with $f(j) := 1, g(j) = \alpha j$ are exactly the $\alpha$-gapped repeats without overlap.
  \citet{Kolpakov17arbitrary} showed that the number of all maximal $(f,g)$-gapped repeats is bounded by
  \[
  \Oh{n \tuple{1 + 
  \max\tuple{%
	  \sup_{j \in \N} (1/j) (f(j) - g(j))
	  , \sup_{j \in \N} \abs{f(j+1)-f(j)}
	  , \sup_{j \in \N} \abs{g(j+1)-g(j)}
  }}}.\]
  Shaping the upper bound, or devising a lower bound for certain $f$ and $g$ is left for future work.

Regarding the algorithmic part,
  \citet{brodal99gap} presented an algorithm computing all maximal $(f,g)$-gapped repeats in \Oh{n \lg n + \occ} time, where $\occ$ is the number of occurrences.
  In the light that we achieved $\Oh{\alpha n}$ running time for finding all maximal $\alpha$-gapped repeats, 
  it looks feasible to devise an algorithm whose running time depends linearly on $n$ and on the values of $f$ and $g$.
  Needless to say, $(f,g)$-gapped palindromes are also an unexplored topic.

\block{Online Algorithm}
To the best of our knowledge, there has not yet been an algorithm devised for 
computing all maximal $\alpha$-gapped repeats/palindromes of a given word \emph{online}.
We are aware of the algorithm of \citet{Fujishige16gapped} finding all gapped palindromes with a fixed gap 
($\ibeg{\RR\arm}-\iend{\LL\arm}-1 = c$ for a constant~$c$) in \Oh{n \lg \sigma} time online while taking \Oh{n} words of working space.

\block{Distinct Sets}
From literature it is already known that searching all distinct squares~\cite{squares17,CrochemoreIKRRW14} or all distinct ordinary palindromes~\cite{Groult10palindrome} of a word of length~$n$ can be done in \Oh{n} time.
A natural extension is computing all distinct $\alpha$-gapped repeats/palindromes, for which we are unaware of any results, both on 
the combinatorial (like giving an upper bound on the number of all distinct $\alpha$-gapped repeats/palindromes) and on the algorithmic aspects.

\bibliography{papermeta/literature}

\clearpage
\appendix

\section{Missing Proofs}

Here, we show that our bounds obtained in \Cref{thmMaxReps} hold when supporting overlaps as we do.
\Cref{thmMaxReps} uses results of~\cite{gawrychowski18tighter}, where gapped repeats are divided into $\beta$-periodic and $\beta$-aperiodic gapped repeats.
Lemma~9 in~\cite{gawrychowski18tighter} for the maximal $\alpha$-gapped $\beta$-aperiodic repeats does not assume that $\iend{\LL\arm} < \ibeg{\RR\arm}$, and therefore supports gapped repeats with overlap.
It is left to show a slightly modified proof of~\cite[Lemma 8]{gawrychowski18tighter}, which treats the maximal $\alpha$-gapped $\beta$-periodic repeats:
\begin{lemma}\label{lemmaPRepeats}
Let $\warm$ be a word, $\alpha > 1$ and $0 < \beta < 1$ 
two real numbers.
Then the number of maximal $\alpha$-gapped $\beta$-periodic is at most $2 \alpha \sumExp{\warm} / \beta$.
\end{lemma}
\begin{proof}
	Let $(\LL\arm,\RR\arm)$ be a maximal $\alpha$-gapped $\beta$-periodic repeat, 
  $\qvar := \ibeg{\RR\arm} - \ibeg{\LL\arm}$ its period,
	and
	$\Sarm := \abs{\LL\arm} = \abs{\RR\arm}$ the length of its arms.
  By definition, the left arm~$\LL\arm$ has a periodic prefix~$\LL\subarm$ of length at least $\beta \Sarm$.
  Let $\LL\rep$ denote the run that generates $\LL\subarm$, i.e.,
  $\LL\subarm \subseteq \LL\rep$. 
  The two segments~$\LL\subarm$ and~$\LL\rep$ have the shortest period~$\period$ in common.
  By the definition of the gapped repeats, there is a right copy $\RR\subarm$ of $\LL\subarm$ contained in $\RR\arm$ with
  $\RR\subarm \equiv \substr{\warm}{\ibeg{\LL\subarm}+\qvar}{\iend{\LL\subarm}+\qvar} = \LL\subarm$.
  Let $\RR\rep$ be a run generating $\RR\subarm$ (it is possible that $\RR\rep$ and $\LL\rep$ are identical).
  By definition, $\RR\rep$ has the same period~$\period$ as $\LL\rep$.

  Since $\agr$~is maximal, $\ibeg{\LL\arm} = \ibeg{\LL\rep}$ or $\ibeg{\RR\arm} = \ibeg{\RR\rep}$ must hold~(see \Cref{figPeriodic});
  otherwise we could extend $\agr$ to the left.

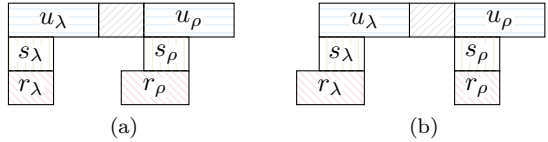
\begin{figure}[ht]
	\centering{%
\subfloat[ ]{%
		\begin{tikzpicture}[yscale=0.45,xscale=0.6]
  \col[arm]{2}{$\LL\arm$}
  \col[gap]{1}{}
  \col[arm]{2}{$\RR\arm$}
  \newrow
  \col[subarm]{1}{$\LL\subarm$}
  \void{2}
  \col[subarm]{1}{$\RR\subarm$}
  \newrow
  \col[rep]{1}{$\LL\rep$}
  \void{1.5}
  \col[rep]{1.5}{$\RR\rep$}
	\end{tikzpicture}
}%
\hspace{2em}%
\subfloat[ ]{%
		\begin{tikzpicture}[yscale=0.45,xscale=0.6]
  \col[arm]{2}{$\LL\arm$}
  \col[gap]{1}{}
  \col[arm]{2}{$\RR\arm$}
  \newrow
  \col[subarm]{1}{$\LL\subarm$}
  \void{2}
  \col[subarm]{1}{$\RR\subarm$}
  \newrow
  \void{-0.5}
  \col[rep]{1.5}{$\LL\rep$}
  \void{2}
  \col[rep]{1}{$\RR\rep$}
	\end{tikzpicture}
}%
}%
\caption{%
	Setting of the proof of \Cref{lemmaPRepeats}.
	Each figure shows a maximal $\alpha$-gapped $\beta$-periodic repeat~$(\LL\arm,\RR\arm)$ and the periodic prefixes~$\LL\subarm$ and $\RR\subarm$
	of its respective arms~$\LL\arm$ and $\RR\arm$. The periodic prefixes are contained respectively in the runs~$\LL\rep$ and $\RR\rep$.
	The equation~\ref{itPeriodicA} $\ibeg{\LL\arm} = \ibeg{\LL\rep}$ or~\ref{itPeriodicB} $\ibeg{\RR\arm} = \ibeg{\RR\rep}$ must hold.
	By the maximality property of runs, $\iend{\LL\rep} = \iend{\LL\subarm}$ and $\iend{\RR\rep} = \iend{\RR\subarm}$,
	i.e., $\LL\subarm \equiv \LL\rep \cap \LL\arm$ and $\RR\subarm \equiv \RR\rep \cap \RR\arm$.
  }
\label{figPeriodic}
\end{figure}

\begin{figure}[ht]
	\MyFloatBox{%
		\begin{tikzpicture}[yscale=0.45,xscale=0.4]
  \col[rep]{2.5}{$\LL\rep$}
  \void{1}
  \col[rep]{3}{$\RR\rep$}
  \newrow
  \col[subarm]{2.5}{$\LL\subarm$}
  \void{1.5}
  \col[subarm]{2.5}{$\RR\subarm$}
  \newrow
  \col[arm]{3}{$\LL\arm$}
  \col[gap]{1}{}
  \col[arm]{3}{$\RR\arm$}
  \newrow
  \col[arm]{1.5}{$\s{\LL\arm}$}
  \col[gap]{3.5}{}
  \col[arm]{1.5}{$\s{\RR\arm}$}
  \newrow
  \col[subarm]{1.5}{$\s{\LL\subarm}$}
  \void{3.5}
  \col[subarm]{1.5}{$\s{\RR\subarm}$}
  \newrow
  \newrow
  \void{4}
  \dist[$\delta \in \period \N$]{1}
	\end{tikzpicture}
	\hspace{1em}
		\begin{tikzpicture}[yscale=0.45,xscale=0.4]
  \col[rep]{2}{$\LL\rep$}
  \void{1.5}
  \col[rep]{2.5}{$\RR\rep$}
  \newrow
  \col[subarm]{2}{$\LL\subarm$}
  \void{2}
  \col[subarm]{2}{$\RR\subarm$}
  \newrow
  \col[arm]{3}{$\LL\arm$}
  \col[gap]{1}{}
  \col[arm]{3}{$\RR\arm$}
  \newrow
  \col[arm]{2.5}{$\s{\LL\arm}$}
  \col[gap]{3}{}
  \col[arm]{2.5}{$\s{\RR\arm}$}
  \newrow
  \void{5}
  \col[rep]{2.5}{$\s{\RR\rep}$}
  \newrow
  \col[subarm]{2}{$\s{\LL\subarm}$}
  \void{3.5}
  \col[subarm]{2}{$\s{\RR\subarm}$}
  \newrow
  \newrow
  \void{5}
  \dist[$< \period$]{1}
	\end{tikzpicture}
		}{%
\caption{%
	Setting of the proof of Case~\ref{itPeriodicA} in \Cref{lemmaPRepeats} for two different maximal $\alpha$-gapped repeats~$(\LL\arm,\RR\arm)$ and $(\s{\LL\arm},\s{\RR\arm})$
with $\ibeg{\LL\arm} = \ibeg{\s{\LL\arm}} = \ibeg{\LL\rep}$.
\emph{Left:} The periodic prefixes~$\RR\subarm$ and $\s{\RR\subarm}$ of the right arms of both gapped repeats are contained in a single run.
	The minimal period~$\period$ of both runs~$\LL\rep$ and $\RR\rep$ determine the possible starting positions of the right arms.
	\emph{Right:} The periodic prefixes of the right arms of both gapped repeats are contained in different runs.
	Both runs cannot overlap more than $\period-1$ positions due to \Cref{lemmaRepetitiveOccs}.
  }
\label{figPeriodicOverlap}
		}%
\end{figure}
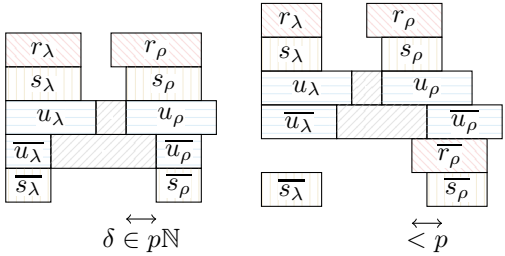

  The periodic $\alpha$-gapped repeat~$\agr$ is uniquely determined by
  its period~$\qvar$ and
  \begin{enumerate}[(a)]
	  \item $\LL\rep$ in case $\ibeg{\LL\arm} = \ibeg{\LL\rep}$, or \label{itPeriodicA}
	  \item $\RR\rep$ in case $\ibeg{\RR\arm} = \ibeg{\RR\rep}$. \label{itPeriodicB}
	  \end{enumerate}
	  Since $\agr$ is $\alpha$-gapped, it holds that $\qvar \le \alpha \Sarm$.
	  We analyze Case~\ref{itPeriodicA}, where $\ibeg{\LL\arm} = \ibeg{\LL\subarm} = \ibeg{\LL\rep}$ holds. 
	  Case~\ref{itPeriodicB} is treated exactly in the same way by symmetry.
    The gapped repeat~$\agr$ is identified by its period~$\qvar$ and~$\LL\rep$.
  We fix $\LL\rep$ and pose the question how many maximal periodic gapped repeats can be generated by $\LL\rep$.
  We answer this question by counting the number of possible values for the period~$\qvar$.
  Since the starting position $\ibeg{\RR\subarm} = \ibeg{\RR\arm} = \ibeg{\LL\arm} + \qvar = \ibeg{\LL\rep} + \qvar$ of the periodic segment~$\RR\subarm$ 
  is determined by~$\qvar$, two possible values of~$\qvar$ must have a distance of at least~$\period$ due to \Cref{lemmaRepetitiveOccs}, see also \Cref{figPeriodicOverlap}.

  With $\Sarm \le \abs{\LL\subarm} / \beta$ and $\qvar \le \alpha \Sarm$, we obtain
    $1 \le \qvar \leq \abs{\LL\subarm} \alpha / \beta \leq \abs{\LL\rep} \alpha / \beta$.
	Then the number of possible periods~\qvar{} is at most $\abs{\LL\rep} \alpha / (\beta \period) = \exp(\LL\rep) \alpha / \beta$.
	Overall, the number of all maximal $\alpha$-gapped repeats is at most $\alpha \sumExp{\warm} / \beta$ for the case $\ibeg{\LL\arm} = \ibeg{\LL\rep}$.
	Since Case~\ref{itPeriodicB} with $\ibeg{\RR\arm} = \ibeg{\RR\rep}$ is symmetric, we get the total upper bound $2 \alpha \sumExp{\warm} / \beta$.
  \end{proof}

\end{document}